\documentclass[aps,pra,twocolumn,10pt]{revtex4}
\usepackage[utf8]{inputenc}
\usepackage[T1]{fontenc}
\usepackage{graphicx}
\usepackage{dcolumn}
\usepackage{bm}
\usepackage{multirow}
\usepackage{amssymb}
\usepackage{float}
\usepackage{hyperref}
\usepackage{amsfonts}
\usepackage{amsmath}
\usepackage{amsthm}
\usepackage{enumerate}
\usepackage{tikz}
\usetikzlibrary{shapes,arrows}
\usepackage{mathtools}
\usepackage[]{algorithm2e}
\usepackage{xcolor}
\usepackage{marvosym}
\newcommand{\mathbbm}[1]{\text{\usefont{U}{bbm}{m}{n}#1}}

\usepackage{tikz}
\usetikzlibrary{matrix,positioning,decorations.pathreplacing}
\newtheorem{theorem}{Theorem}

\newtheorem{observation}{Observation}
\newtheorem{corollary}{Corollary}[theorem]

\newtheorem{lemma}{Lemma}
\newtheorem{definition}{Definition}
\newcommand{\tr}{\text{tr}}

\newcommand{\id}{\ensuremath{\mathbbm{1}}} 
\newcommand{\one}{\id}

\newcommand{\bra}[1]{\langle #1|}
\newcommand{\ket}[1]{|#1\rangle}

\newcommand{\R}{\ensuremath{\mathbbm R}}
\newcommand{\C}{\ensuremath{\mathbbm C}}
\newcommand{\N}{\ensuremath{\mathbbm N}}
\newcommand{\nlc}{T}
\newcommand{\bi}{\begin{itemize}}
\newcommand{\ei}{\end{itemize}}

\newcommand{\be}{\begin{equation}}
\newcommand{\ee}{\end{equation}}
\newcommand{\bea}{\begin{eqnarray}}
\newcommand{\eea}{\end{eqnarray}}

\begin{document}

\title{Symmetries and entanglement of stabilizer states}

\date{\today}

\author{Matthias Englbrecht, Barbara Kraus}
\affiliation{Institute for Theoretical Physics, University of Innsbruck, A–6020 Innsbruck, Austria}

\begin{abstract}
Stabilizer states constitute a set of pure states which plays a dominant role in quantum error correction, measurement--based quantum computation, and quantum communication. Central in these applications are the local symmetries of these states. We characterize all local symmetries of arbitrary stabilizer states and provide an algorithm which determines them. We demonstrate the usefulness of these results by showing that the additional local symmetries find applications in entanglement theory and quantum error correction.
\end{abstract}

\maketitle

\section{Introduction}
Entanglement has been identified as a crucial property to investigate, describe, and leverage applications in several areas of Science \cite{review, AmFa08}. It is essential for quantum computation \cite{RaBr01} as well as certain quantum communication schemes \cite{SecretSh}. Moreover, in the last decade, concepts developed in entanglement theory have been utilized in other fields of research \cite{AmFa08}. Hence, an enormous effort has been made to qualify and quantify entanglement \cite{review}. Despite extensive investigations in the context of quantum information theory, its detailed characterization and quantifications remain, however, as major challenges.

A set of states which is key in the aforementioned applications within quantum information theory is the set of stabilizer states \cite{HeMa06}. A $n$ qubit stabilizer state is defined as the unique simultaneous eigenvector of a maximal set of commuting operators in the Pauli group, which is defined as the tensor product of either a Pauli operator or the identity operator. These states, which can be highly entangled, are used in quantum error correction \cite{GoDa97}, in measurement based quantum computation \cite{RaBr01}, and in self-testing \cite{Acin} to just name a few applications. Some of the entanglement properties of stabilizer states have been investigated \cite{HeMa06,Cubitt}. Furthermore, purification protocols have been developed \cite{GlKn06}. Stabilizer states also found applications in proving a separation between universal quantum computation and computations which are classically efficiently simulable \cite{GoDa99}. Given that all these applications stem from the rich entanglement capability and from the local symmetries of these states, the further investigation of the both, the entanglement properties and the local symmetries of stabilizer states is indispensible. Arguably, a deep understanding of those characteristics will allow one to identify new applications of multipartite entanglement.

Entanglement theory is a resource theory, where the free operations are Local Operations assisted by Classical Communication (LOCC). LOCC arose as a natural and operationally motivated  choice  of  free  operations, as entanglement  is considered as a  resource  shared  by  different, possibly  spatially separated, parties. Those parties can act locally on their share of the state and can communicate any classical information to the other parties (LOCC), who then manipulate their system accordingly. LOCC extends Local Unitary (LU) transformations, where no communication is allowed and the transformations are restricted to be unitary, and hence, reversible, so that they do not alter entanglement. As LOCC cannot generate entanglement, it holds that, whenever a pure state $\ket{\Psi}$ can be transformed deterministically into some other state $\ket{\Phi}$, then the entanglement of the latter is at most as large as the one of the former. Important to note here is that this holds true for any entanglement measure. Stated differently, a entanglement measure is a functional which is non--increasing under LOCC. Hence, the study of LOCC transformations allows to order the set of entangled states. Despite the fact that LOCC transformations constitute a intricate set of operations \cite{LOCC}, it was recently shown that fully entangled pure states describing $n$ qudits (homogeneous system) can generically not be (deterministically) transformed into any other LU--inequivalent fully entangled pure state \cite{GoKr17,SaGo18}. Hence, the partial order which is obtained from the study of possible LOCC transformations is generically trivial for homogeneous systems. That is, in stark contrast to the bipartite case \cite{nielsen}, almost all states are isolated; they can neither be reached from, nor can they be transformed into another pure LU-inequivalent state. However, those sets of states which play a major role in physics, such as tensor network states \cite{MPS,SaMo19,orus, VeCi04} and the here studied stabilizer states, are always of measure zero, which implies that the results for generic states might not be applicable.
In fact, the results mentioned above are a consequence of a more general result, which states that pure states describing $n$ qudit systems do generically not possess any local symmetry (other than the identity). As mentioned above, stabilizer states do not only possess local symmetries, but are actually defined by their local symmetries. These symmetries ensure that stabilizer states can indeed be transformed into some other LU-inequivalent state via LOCC (see also below). Apart from the relevance of those local symmetries in entanglement theory, they also play an important role in universal measurement based quantum computation as recently shown in \cite{RaOk19}. Furthermore, symmetries beyond the Pauli group are useful in fault tolerant quantum computing \cite{GoDa97, CE17, NC00} (see also Sec. \ref{Local symmetries and transversal gates}). Hence, the identification of additional symmetries has already been subject to several other works \cite{WK11,DZ09,GiGl20}. Here, we go beyond these investigations by providing a complete characterization of all local symmetries of stabilizer states. As mentioned before, these additional symmetries are not only relevant to identify new error correcting codes for which transversal gates exist, but can also be used to shine new light on entanglement theory. In fact, as we will show, the additional symmetries can be used to identify new transformations which are possible via LOCC, hence, leading to a non--trivial order of entanglement. Furthermore, these additional symmetries have been used in \cite{HeMa19} to demonstrate a difference between pure state transformations via separable maps, which utilize singular matrices and those which do not. Here, a completely positive map is called separable if it possesses a decomposition such that all Kraus operators are local operators. Despite the fact that this is a very abstract result, it has far reaching consequences, as it shows that some of the previous investigations concerning LOCC have to be revised \cite{HeMa19}.

The outline of the remaining paper is the following. First, we introduce our notation and review some relevant results in the theory of stabilizer states. The aim of Sec. \ref{Symmetries of stabilizer states} is to characterize all possible invertible local symmetries of stabilizer states. As we will see, one can restrict this investigation to local unitary symmetries, as all other symmetries will be determined by them. We will first show that a stabilizer state possessing infinitely many symmetries has to correspond to a graph which possesses a leaf (see also \cite{WK11,DZ09}). This refers to a particular structure of the underlying graph, which can be easily identified by considering the two-qubit reduced states. All other states only possess finitely many local symmetries. We will then show that any additional symmetry implies the existence of a symmetry which is a local Clifford gate. In Sec. \ref{clifford sym} we will then derive necessary and sufficient conditions for the existence of local Clifford symmetries. An algorithm to determine then all local (invertible) symmetries of an arbitrary stabilizer state is presented in Sec. \ref{alg all sym}. In Sec. \ref{examples} we provide examples of states possessing non-trivial symmetries. We illustrate the usefulness of these results by utilizing them in various contexts in Sec. \ref{applications}. First, we demonstrate that the additional symmetries can be employed to identify error correcting codes, which possess transversal gates \cite{CE17}.  In the context of entanglement theory, the characterization of the local symmetries of stabilizer states presented here, will be used to provide a general construction for separable maps among pure states which are more general than what was previously considered. The consequences of this result within entanglement theory are explained in \cite{HeMa19}. As a final application we will show that states with additional symmetry are indeed more powerful in the sense that they can reach more states via LOCC (where we consider realistic LOCC protocols which utilize finitely many rounds of classical communication).  In Sec. \ref{Conclusion and Outlook} we conclude and discuss future research directions.

\section{Notation \& Preliminaries\label{not and pre}}

First let us introduce the notation and recall some results concerning stabilizer states and their symmetries.

\subsection{Notation}
In the following we denote the Pauli operators by $X,Y,Z$ and the identity operator by $\one$. Whenever we consider a state of $n$ qubits the usage of a single subscript for an operator denotes the system the operator is acting on unless stated otherwise. If an operator acting on a qubit has two indices $O_k^j$, $k\in \N$, then the superscript, $j$, denotes the qubit the operator is acting on and the subscript, $k$, labels different operators acting on qubit $j$. Furthermore, if $O$ is a local operator acting on $n$ qubits we denote by $\text{supp}(O)$ the support of $O=O_1\otimes\ldots\otimes O_n$, i.e. the subset of qubits $j$ for which $O_j\not\propto\one$.

In this paper we determine the local symmetry group of a general stabilizer state $\ket{\psi}\in(\C^2)^{\otimes n}$. We will denote this group by
\begin{equation}
G_\psi=\{G=G_1\otimes \ldots\otimes G_n\in GL(2)^{\otimes n}\ |\ G\ket{\psi}\propto\ket{\psi}\}.
\label{eq:gl}
\end{equation}
A subgroup of this local symmetry group consists of the local unitary symmetries of $\ket{\psi}$ which we will refer to as
\begin{equation}
U_\psi=\{U=U_1\otimes \ldots\otimes U_n\in U(2)^{\otimes n}\ |\ U\ket{\psi}\propto\ket{\psi}\}.
\label{eq:lu}
\end{equation}
We will see that in order to determine $G_\psi$ for a stabilizer state it is sufficient to restrict to graph states \cite{HeMa06}, a special type of stabilizer state defined by a mathematical graph $G(V,E)$. Here $V$ is the set of vertices of the graph and $E\subseteq V\times V$ is the set of edges. In the following we consider only simple, undirected graphs, i.e. graphs without self-loops and double-edges and direction of the edges. Such a graph is in 1 to 1 correspondence with a symmetric binary matrix called adjacency matrix. For a graph $G=(V,E)$ the adjacency matrix $\theta\in \{0,1\}^{|V|\times|V|}$ is defined by $\theta_{ij}=1$ for $\{i,j\}\in E$ and $0$ otherwise. For a simple graph this matrix is symmetric and has zeros on the diagonal. Let us now introduce some graph theoretic terms which will become relevant later on.
\begin{definition}[neighbourhood]
Let $G=(V,E)$ be a graph. Then the neighbourhood of vertex $j\in V$, $N_j$, is the set of vertices adjacent to $j$, i.e.
\begin{equation}
N_j=\{k\in V|\{j,k\}\in E\}.
\end{equation}
\end{definition}
Three graph structures will become important below, leaf and parent, twin vertices and connected twin vertices. A vertex $l\in V$ is called leaf if it is connected exactly to one other vertex, i.e. $|N_l|=1$. The vertex it is connected to $p\in N_l$ is called its parent. Moreover, two vertices $s,t\in V$ are called twins if they have the same neighborhood, i.e. $N_s=N_t$. They are called connected twins if they share all neighbours and are connected, i.e. $N_s\backslash\{t\}=N_t\backslash\{s\}$ and $s\in N_t$.

Furthermore, we make use of two groups in the following, the Pauli group and the local Clifford group. We denote by $\mathcal{P}_1=\left<X,Y,Z\right>$ the one qubit Pauli group. The $n$-qubit Pauli group $\mathcal{P}_n$ is given by
\begin{equation}
\mathcal{P}_n=\left<\{\sigma_1\otimes\ldots\otimes\sigma_n|\sigma_1,\ldots,\sigma_n\in \mathcal{P}_1\}\right>.
\label{eq:paulin}
\end{equation}
The (local) Clifford group for one qubit is defined as
\begin{equation}
\mathcal{C}_1=\{U\in U(2)|\forall \sigma\in\mathcal{P}: U\sigma U^\dagger\in \mathcal{P}\}
\label{eq:clifford1}
\end{equation}
and the $n$ qubit local Clifford group $\mathcal{C}_n$ is defined as the group generated by the $n$-fold tensor products of elements of $\mathcal{C}_1$, i.e.
\begin{equation}
\mathcal{C}_n=\left<\{c_1\otimes\ldots\otimes c_n|c_1,\ldots,c_n\in\mathcal{C}_1\}\right>.
\end{equation}
The group $\mathcal{C}_1$ has $24$ elements (up to phases) \footnote{We do not need to consider these phases as we do identify all symmetries in $GL$ (see Eq. (\ref{eq:gl}) and (\ref{eq:lu}))}. The factor group $\mathcal{C}_1/\mathcal{P}_1$ has $6$ elements and is isomorphic to the symmetric group of $3$ elements. The elements of $\mathcal{C}_1$ are of the form
\begin{align}
\sigma_1\exp\left(i\alpha\sigma_2\right),\alpha\in\{0,\pm\frac{\pi}{4}\}\label{eq:formc1even}\\
\sigma_1\exp\left(i\beta\sigma_2\right)\exp\left(i\gamma\sigma_3\right),\beta,\gamma\in\{\pm\frac{\pi}{4}\}\label{eq:formc1odd},
\end{align}
where $\sigma_1\in\{\one,X,Y,Z\},\sigma_2, \sigma_3 \in\{X,Y,Z\}$ with $\sigma_2\neq \sigma_3$. We mention the specific form of the elements since it is related to the order of the element. The order of a group element $c$ is defined as the smallest integer $k\in\mathbb{N}$ such that $c^k=\one$. Elements of the form in Eq. (\ref{eq:formc1even}) have order $1,2$ or $4$ and elements of the form in Eq. (\ref{eq:formc1odd}) have order $3$. As the latter subset plays an important role in the following, we will use the notation $\mathcal{C}^3_1=\{U\in \mathcal{C}_1\setminus \one | U^3=\one\}$. Furthermore, we abbreviate local Clifford operators, i.e. operators in $\mathcal{C}_n$, by LC.

\subsection{Graph states and stabilizer states\label{Graphs and graph states}}
A stabilizer state is defined as follows. Let $\mathcal{S}_\psi=\left<g_1,\ldots,g_n\right>\subset\mathcal{P}_n$ be an abelian subgroup of the $n$-qubit Pauli group, generated by $n$ independent elements $g_1,\ldots,g_n\in\mathcal{P}_n$ with $-\one\not\in \mathcal{S}_\psi$. The state $\ket{\psi}$, which is the unique eigenstate for all elements of $\mathcal{S}_\psi$ to eigenvalue $+1$, is called stabilizer state and $\mathcal{S}_\psi$ is called its stabilizer. Clearly we have $\mathcal{S}_\psi\subseteq U_\psi\subseteq G_\psi$.

Every stabilizer state is LC equivalent to a graph state \cite{NM04,HeMa06}. A graph state is defined via a mathematical graph. For a graph $G=(V,E)$ we denote by $\ket{G}\in\left(\mathbb{C}^2\right)^{\otimes n}$ the corresponding graph state defined as
\begin{equation}
\ket{G}=\prod_{\{i,j\}\in E}(U_{cz})_{ij} \ket{+}^{\otimes n}.
\label{eq:graphstate}
\end{equation}
Here $U_{cz}$ denotes the controlled phase gate. The state defined this way is a stabilizer state and its stabilizer $\mathcal{S}_G$ is generated by the operators $S_{(j)}=X_j\bigotimes_{k\in N_j} Z_{k}$, $j\in\{1,\ldots,n\}$ (canonical generators). Without loss of generality we consider only fully entangled stabilizer states in the following, i.e. states which correspond to fully connected graphs.

\subsubsection{Local complementation}
There exists a close connection between the local Clifford group and an operation on graphs called local complementation which we briefly review here \cite{NM04}. Given a graph $G=(V,E)$, local complementation at a vertex $v\in V$ yields a new graph $G'=(V',E')$ with $V'=V$ and $E'=E\oplus \{\{a,b\}|a,b\in N_v, a\neq b\}$ where $\oplus$ is the exclusive OR for sets. In other words local complementation at vertex $v$ complements  the subgraph spanned by the neighbourhood of $v$. The corresponding graph states $\ket{G}$ and $\ket{G'}$ are related by an LC operation \cite{NM04} as follows

\bea
\ket{G'}=e^{-i\frac{\pi}{4}X_v}\otimes\bigotimes_{j\in N_v} e^{i\frac{\pi}{4}Z_j}\ket{G}.\label{eq:tttemp1}
\eea

One can show that a graph state $\ket{G_1}$ is LC equivalent to another graph state $\ket{G_2}$, i.e. there exists a $C\in\mathcal{C}_n$ such that $\ket{G_1}=C\ket{G_2}$, iff $G_1$ is equivalent to $G_2$ up to a sequence of local complementations \cite{NM04}. As we will show in the following a leaf in a graph is associated with an additional symmetry of the corresponding graph state. Clearly, any graph obtained by local complementation also possesses an additional symmetry. As was shown in \cite{AB93}, using local complementation one can switch the role of leaf and parent, turn leaf and parent into twin vertices or into connected twin vertices. Note that this is the whole orbit of a leaf parent pair under local complementation.

Before we study the additional symmetries of graph states let us make some simple observations. First note that $e^{i\alpha S}\ket{G}\propto\ket{G}$ holds for any $\alpha\in \mathbb{R}$ and $S\in\mathcal{S}_G$. However, it can be easily seen that such an operator is never local. Next we have that $\exp({i\alpha X_j})\ket{G}=\exp(i\alpha \bigotimes_{k\in N_j}Z_k)\ket{G}$ for any qubit $j$ where the operator $\exp(i\alpha \bigotimes_{k\in N_j}Z_k)$ is in general (unless $|N_j|=1$, i.e. vertex $j$ is a leaf) non-local. Let us remark her that using local complementation one can see that for $\alpha=\pi/4$ we have $\exp(i\alpha X_j)\ket{G}=\bigotimes_{k\in N_j} \exp(i\alpha Z_k)\prod_{m\neq n\in N_j} (U_{cz})_{mn}\ket{G}$ (s. Eq. (\ref{eq:tttemp1})).

\subsection{Additional symmetries of stabilizer states}

We investigate here all symmetries of stabilizer states. That is, we characterize the group $G_\psi$ for a general stabilizer state $\ket{\psi}\in\left(\mathbb{C}^2\right)^{\otimes n}$. This characterization is simplified by the following two observations. First, as mentioned above, every stabilizer state is LC equivalent to a graph state \cite{NM04}. Hence, we can restrict the considerations to graph states. Second, any graph state is critical, i.e. all single qubit reduced states are proportional to the identity. For these states local invertible symmetries result from local unitary ones \cite{NW17}. Hence, it is sufficient to characterize local unitary symmetries to determine the whole local symmetry group in $GL^{\otimes n}$.
Thus, the question we have to answer is: What is the form of $U_G$ for a general graph state $\ket{G}\in\left(\mathbb{C}^2\right)^{\otimes n}$?

Let us also mention here that the fact that there always exists only one critical state in an SLOCC--class (up to LUs) and the fact that stabilizer states are critical imply that no pair of stabilizer states is SLOCC--equivalent unless it is LU-equivalent \cite{GoGi11}. Here SLOCC denotes stochastic local operations assisted by classical communication (LOCC). Mathematically speaking, two states $\ket{\psi}$ and $\ket{\phi}$ are SLOCC--equivalent, i.e. they belong to the same SLOCC class, iff there exists a local invertible operator $h=h_1\otimes \ldots\otimes h_n$ such that $\ket{\psi}=h\ket{\phi}$.

The problem of LU--equivalence of stabilizer states has been extensively studied in the literature \cite{NeMa05,DG07,LuLc10}. Here, we recall some results, which are needed subsequently. Let us begin by recalling the definition of semi Clifford operators.

\begin{definition}[semi Clifford]\label{definition semi clifford}
An operator $O\in U(2)$ is called semi Clifford if there exists $\sigma\in\{X,Y,Z\}$ such that $U\sigma U^\dagger \in\mathcal{P}_1$.
\end{definition}
Hence, in contrast to  Clifford operators, which leave the whole Pauli group invariant, a semi Clifford operator only maps at least one Pauli operator back to the Pauli group. A local operator $O=O_1\otimes \ldots \otimes O_n\in U(2)^{\otimes n}$ is called a local semi Clifford operator if $O_j$ is a semi Clifford operator for all $j\in \{1,\ldots,n\}$. It is straightforward to show (see Appendix A) that $U$ is a semi Clifford operator iff it can be written as
\bea U\propto C e^{i \alpha \sigma},
\label{eqn: form semi clifford}\eea
where $C\in \mathcal{C}_1$, $\alpha \in \R$ and $\sigma\in \{X, Y, Z\}$. Thus, up to local Clifford operators a semi Clifford operator is diagonal \cite{DG07}, i.e. $U=Ce^{i \alpha \sigma}=CEe^{i \alpha Z}E^\dagger$ where $E\in\mathcal{C}_1$ is the Clifford operator mapping $\sigma$ to $Z$. It was shown in \cite{DG07} that LU operators relating stabilizer states have to be local semi Clifford operators as stated in the following theorem.

\begin{theorem}[\cite{DG07}]
Let $\ket{\psi}$ and $\ket{\psi'}$ be fully entangled, LU--equivalent stabilizer states on $n\ge 3$ qubits  and let $U=U_1\otimes\ldots\otimes U_n\in U(2)^{\otimes n}$ be such that $U\ket{\psi}=\ket{\psi'}$. Then $U$ is a local semi Clifford operator.
\label{theorem semi clifford}
\end{theorem}

Let us also make the following simple observations regarding the local symmetries of a stabilizer state.

\begin{observation}
Let $\ket{\psi}$ be a stabilizer state and let $P\in \mathcal{P}_n$ be a symmetry, i.e. $P\in U_\psi$. Then there exists a $\lambda\in\{\pm 1,\pm i\}$ such that $\lambda P\in\mathcal{S}_\psi$.
\label{nootherpauli}
\end{observation}

This can be easily seen by observing that $P$ has to commute with all elements of $S_\psi$. In order to see this, suppose that $P\ket{\psi}=\alpha\ket{\psi}$ with $\alpha\in\mathbb{C}$. If $P$ does not commute with all elements of $S_\psi$ there would exist a $S\in \mathcal{S}_\psi$ such that $SP\ket{\psi}=-PS\ket{\psi}=-\alpha\ket{\psi}$ and $SP\ket{\psi}=\alpha\ket{\psi}$ which leads to a contradiction. As $\mathcal{S}_\psi$ is a maximal set of commuting Pauli operators we conclude that $\alpha^{-1}P\in\mathcal{S}_\psi$.

\begin{observation} \label{2 factors}
There exists no local symmetry of a (fully entangled) stabilizer state which acts non--trivially only on one qubit.
\end{observation}

This observation can be easily proven by noting that $\ket{\psi}=\ket{0}\ket{\psi_0}+\ket{1}\ket{\psi_1}$, where $\langle \psi_1\ket{\psi_0}=0$. Hence, $U_1 \otimes \one \ket{\psi}\propto \ket{\psi}$ iff $U_1$ is proportional to the identity.

\section{Symmetries of stabilizer states\label{Symmetries of stabilizer states}}

Using Theorem \ref{theorem semi clifford} we first derive necessary conditions on the local symmetries of an arbitrary stabilizer state. More precisely, we show that any $U_j$ being part of a symmetry of a stabilizer state has to be either an LC operator of order $3$, or of some other particular form. We will then study these two cases separately and will derive necessary conditions for the existence of these symmetries. As we will see, unless the graph state contains a leaf, only discrete symmetries exist. Note that this was already shown in \cite{DZ09,WK11}. Furthermore, we will show that all other symmetries can be derived by characterizing the local Clifford symmetries of graph states, for which we present necessary and sufficient, easily computable, conditions. Finally, we will present an algorithm which determines all symmetries of an arbitrary graph (stabilizer) state.

Let us start out by characterizing the local unitary symmetries for an arbitrary graph state. As shown in the following theorem, these symmetries can be constraint to a very special form.

\begin{theorem}
Let $\ket{\psi}\in\left(\mathbb{C}^2\right)^{\otimes n}$ be a fully entangled stabilizer state and let $U\in U_\psi$ be a local symmetry of $\ket{\psi}$. Then

\begin{equation}
U_j\propto \left\{\begin{matrix}C_j \\\sigma_1^j\exp\left(i\alpha_j\sigma_2^j\right)
\end{matrix}\right.
\label{eq:generalform}
\end{equation}

with $C_j\in\mathcal{C}_1^3$, $\alpha_j\in\mathbb{R}$, $\sigma_1^j\in\{\one,X,Y,Z\}$ and $\sigma_2^j\in\{X,Y,Z\}$ for all $j\in\{1,\ldots,n\}$.
\label{generalform}
\end{theorem}

Hence, any unitary which could potentially occur as a tensor factor of a local symmetry of a stabilizer state must be either a local Clifford operator of order $3$, or of the form $\sigma_1^j\exp(i\alpha_j\sigma_2^j)$. Note that the latter form includes all other Clifford operators. To show Theorem \ref{generalform} we make use of Theorem \ref{theorem semi clifford} which implies that any symmetry $U\in U_G$ has to be a local semi Clifford operator. As $U_G$ forms a group, the same has to hold for $U^2$. Using this it is straightforward to derive Theorem \ref{generalform}. For a detailed proof we refer the reader to Appendix \ref{appendixb}.

In the subsequent sections we will derive necessary conditions on the existence of those symmetries. To this end, we will first consider graph states for which all $U\in U_G$ are such that all $U_j\propto \sigma_1^j\exp(i\alpha_j\sigma_2^j)$. Let us denote this set of graph states by
\bea
\begin{split}&\nlc=\\ &\left\{\ket{G}| \forall U \in U_G, U=\otimes_k U_k,\text{ where }U_k\notin \mathcal{C}_1^3 \forall k\right \}.\end{split}\label{setT}
\eea
Note that $T$ also contains all graph states with no additional symmetries ($\mathcal{S}_G=U_G$). It will become clear later on that the set of all graph states can be divided into $T$ and a set of graph states with LC symmetries of order $3$. The reason for that is that as soon as one of the tensor factors $U_j$ in a local symmetry $U$ is a Clifford operator of order $3$, then all other $U_k$ must also be elements of $\mathcal{C}_1^3$ (see Lemma \ref{LC3oneLC3all}), i.e. the graph state has a local symmetry in which every tensor factor is an LC of order $3$. Furthermore, combining this with Theorem \ref{generalform} and the group properties of the local symmetry group we find that graph states not in $T$ can only have LC symmetries. While graph states in $T$ can have LC symmetries as well (however only of order $4$), more general local symmetries are possible for these states. 

\subsection{Symmetries for graph states in $\nlc$}

Let us first investigate under which conditions a graph state $\ket{G}\in\nlc$ can have additional symmetries. Recall that for these graph states any $U\in U_G$ is such that $U_j\propto \sigma_1^j\exp(i\alpha_j\sigma_2^j)$ for any $j$. Note that for a given graph state $\ket{G}\in T$ there exists only one $\sigma_2^j\in\{X,Y,Z\}$ for all $U\in U_G$ with $U_j \propto \sigma_1^j\exp(i\alpha_j\sigma_2^j)$ and $\alpha_j\neq \pi/2+k\pi$, $k\in \mathbb{Z}$. This is again a consequence of Theorem \ref{generalform} and the group properties of $U_G$.

We will show that either the graph contains a leaf (up to local complementation), or the phases $\alpha_j$ can only take values $\pi m/2^n$ for some $m,n\in\mathbb{N}$ (wlog we assume that $m<2^n$). As we will see later on (Sec. \ref{Local symmetries and transversal gates}) this statement also follows from \cite{AJ16}. Let us first show that $\alpha_j$ can be different from $\pi m/2^n$ only if qubit $j$ is a leaf (up to local complementation), as stated in the following theorem. Note that here and in the following we consider the respective graph state up to permutations of the qubits. Hence, if we consider a single qubit, we can choose qubit $1$.

\begin{theorem}
Let $\ket{G}\in\nlc$ be a graph state on $n$ qubits and let $U\in U_G$ with $U_1\propto\sigma_{1}^1 \exp({i\alpha_1 \sigma_{2}^1})$ be such that $\alpha_1\neq \frac{m_1\pi}{2^{n_1}}$ for any $m_1, n_1\in\mathbb{N}$. Then the vertex $1$ is a leaf up to local complementation.
\label{thm leaf}
\end{theorem}

\begin{proof}
Let $U\in U_G$ be such that $U_1\propto\sigma_{1}^1\exp({i\alpha_1 \sigma_{2}^1})$ and $\alpha_1\neq \frac{m_1\pi}{2^{n_1}}$ for any $m_1, n_1\in\mathbb{N}$. Since $\ket{G}\in \nlc$ also $U_j \propto\sigma_{1}^j\exp({i\alpha_j \sigma_{2}^j})$ for $j\in\{2,\ldots,n\}$. Note that for any $V\in U_G$ it holds that $S V\in U_G$ for any $S\in \mathcal{S}_G$, and that $V^2,V^\dagger\in U_G$. We use these properties to construct out of $U$ a new local symmetry of $\ket{G}$ (or of a graph state that is LC equivalent to $\ket{G}$). This symmetry acts nontrivial only on qubit $1$ unless vertex $1$ is a leaf up to local complementation. Due to Observation \ref{2 factors}, we conclude that vertex $1$ is a leaf up to local complementation.

Wlog we have $\sigma_{1}^j\neq\sigma_{2}^j$ for all $j\in\{1,\ldots,n\}$. Moreover, if $\sigma_1^1\not\propto \one$ we consider $SU$ instead of $U$ where $S\in\mathcal{S}_G$ is chosen such that $S_1=\sigma_1^1$. This is always possible as we consider fully connected graphs. Furthermore, let us show that it is sufficient to restrict ourselves to the case where $\sigma_2^1=Z$. This follows from that fact that any other $\sigma_2^1$ can be transformed into $Z$ via local complementation. To be more precise, in case $\sigma_2^1=X$ local complementation at any qubit in $N_1$ followed by local complementation at qubit $1$ leads to a new graph state for which the symmetry corresponding to $U$ satisfies $\sigma_2^j=Z$. In case $\sigma_2^1=Y$ this is achieved by local complementation on qubit $1$. The symmetry of the new graph is related to $U$ by conjugation with LC operators. Note that this does not change the phases $\alpha_j$. Hence, wlog we consider the case $\sigma_2^1=Z$ \footnote{If instead we do not use local complementation and consider the three cases $\sigma_2^1=X,Y,Z$ separately we identify the structures corresponding to a leaf under local complementation, i.e. twin vertices and connected twins.}.

Thus, it remains to show the following. If $U_1=e^{i \alpha_1 Z}$ with $\alpha_1\neq \frac{m_1\pi}{2^{n_1}}$ for any $m_1,n_1\in\mathbb{N}$ is a local tensor factor of a unitary symmetry $U$ of a graph state, then vertex $1$ is a leaf (up to local complementation). To show this let us consider the new local symmetry of this graph state $A=U^2S_{(1)}((U)^2)^\dagger S_{(1)}\in U_G$, where $S_{(1)}\in\mathcal{S}_G$ is the canonical generator corresponding to qubit $1$. Observe that for $\sigma_k,\sigma_j \in \{X,Y,Z\}$ we have that $(\sigma_k e^{i\alpha  \sigma_{j}})^2=\one$ if $\sigma_k\neq \sigma_j$ and $(\sigma_k e^{i\alpha  \sigma_{j}})^2=e^{i2\alpha  \sigma_{j}}$ if $\sigma_k= \sigma_j$. Thus, the symmetry $A$ satisfies $\text{supp}(A)\subseteq N_{1}\cup \{1\}$ and $A_1=e^{4i\alpha_1 Z}\not \propto \one$ as $\alpha_1\neq m\pi/2^n$. Furthermore, $A_j\propto e^{i\beta_j \sigma_{2}^j}$ with $\beta_j=0$ if $\sigma_2^j=Z$ or $\sigma_1^j\not \propto\one$ and $\beta_j=4\alpha_j$ if $\sigma_2^j\in\{X,Y\}$ and $\sigma_1^j\propto\one$ for all $j\in N_1$. Let $B$ denote the symmetry we obtain by multiplying $A$ with $S_{(j)}$ and squaring the result subsequently for every qubit $j\in N_1$ with $\beta_j\neq 0$ and $\sigma_2^j=Y$. Note that $(S_{(j)})_1=Z$ for all $j\in N_1$ and hence, $B_1\not\propto \one$. Note further that
$\text{supp}(B)\subseteq N_{1}\cup\{1\}$ and that 
if $B_j\not\propto \one$ then $\sigma_{2}^j=X$ must hold. For any qubit $j\in N_1$ which has a neighbour $k\in N_j$ different from qubit $1$ and for which $B_j$ is nontrivial, we multiply $B$ with $S_{(k)}$ and square the result. Since $k\neq 1$ we have $(S_{(k)})_1=\one$ or $Z$. Let $B'$ be the symmetry obtained in this process. By construction $(B')_1\not\propto \one$ and $(B')_j\propto \one$  $\forall j: N_j\backslash \{1\}\neq \emptyset$. Due to Observation \ref{2 factors} we conclude that there has to exist at least one qubit $j\in N_1$ such that $N_j=\{1\}$. Hence performing local complementation at qubit $1$ followed by a qubit $j$ with $N_j=\{1\}$ turns qubit $1$ into a leaf, which implies the assertion.
\end{proof}

Applying the reasoning of this proof to a vertex $j$ which is not a leaf (under local complementation) we can derive bounds on $n_j$ where $\alpha_j=m_j\pi/2^{n_j}$, as shown in the following.

\begin{corollary}
Let $\ket{G}\in\nlc$ be a graph state on $n$ qubits and let $U\in U_G$ be such that $U_1\propto \sigma_1^1\exp(i\alpha_1\sigma_2^1)$. Then, if vertex $1$ is not a leaf under local complementation it holds that $\alpha_1=m_1\pi/2^{n_1}$ with $m_1,n_1\in\mathbb{N}$ and $n_1\le |N_1|+2$ if $\sigma_2^1\in\{Z,Y\}$ and $n_1\le \text{min}_{j\in N_1}|N_j|+2$ if $\sigma_2^1=X$ (for $|m_1|<2^{n_1}$.
\label{bounds}
\end{corollary}
\begin{proof}
If vertex $1$ is not a leaf under local complementation then by Theorem \ref{thm leaf} we have $U_1\propto \sigma_1^1\exp(i\alpha_1\sigma_2^1)$ with $\alpha_1=m_\pi/2^{n_1}$ and $m_1 ,n_1\in\mathbb{N}$. Let us again construct the symmetry $B'$ from the proof to Theorem \ref{thm leaf} for vertex $1$. Since vertex $1$ is not a leaf under local complementation we have that $(B')_j\propto\one$ for all $j\neq 1$. Using again Observation \ref{2 factors} we conclude that $(B')_1\propto \one$. Counting the number of times we had to square $U_1$ (and $U_1^\dagger$) in the worst case to get $B'$ leads to the stated bounds.
\end{proof}

Theorem \ref{thm leaf} shows that symmetries with phases $\alpha\neq m\pi/2^n$ can only exist in case the graph possesses (up to local complementation) a leaf. In this case, the phase can indeed be arbitrary, as stated in the following observation.

\begin{observation}
Let $\ket{G}$ be a graph state on $n\ge 3$ qubits. Let qubit $1$ and $2$ be a leaf parent pair. Then

\begin{equation}
U=e^{i\alpha X}\otimes e^{-i\alpha Z}\otimes \one\ \ \ \alpha\in\mathbb{R}
\end{equation}

is in $U_G$. Moreover, there exists no other unitary symmetry of the form $U_1\otimes U_2\otimes \one\in U_G$.
\label{char leaf}
\end{observation}

Whereas this result has already been derived in several other works (see \cite{DZ09,WK11}) we present a different proof in Appendix \ref{appendixd}. If qubit $1$ and $2$ are a leaf parent pair only up to local complementation it was shown in \cite{AB93} that the only possible structures are twin vertices or connected twins. It is easy to see that if the qubits are twin vertices then $U=e^{i\alpha X}\otimes e^{-i\alpha X}\otimes \one$, if they are connected twins $U=e^{i\alpha Y}\otimes e^{-i\alpha Y}\otimes \one$, which are of course LC equivalent.

We call the unitary symmetry group resulting from the existence of leaf in the following leaf symmetry and denote the group generated by all leaf symmetries of a graph by $L_G$. Let us now investigate those symmetries which do not stem from a leaf--symmetry. To this end we consider the factor group $U_G / L_G$. Note that this is possible as $L_G$ is a normal subgroup of $U_G$. For a subgroup to be normal it has to be invariant under conjugation by all group elements ($U_G$). For $L_G$ this property can be shown as follows. For any leaf parent pair $(j,k)$ the subgroup $L_G^{(j,k)}=\{\exp(i\alpha X_j)\otimes \exp(-i\alpha Z_k)|\alpha\in\mathbb{R}\}$ generated by the respective leaf symmetry is normal. This holds as the conjugation of any element of $L_G^{(j,k)}$ by a $U\in U_G$ does not change the support of the respective element. Thus, according to Observation \ref{char leaf}, the resulting symmetry has to come from the leaf symmetry corresponding to the leaf parent pair $(j,k)$ and thus is again an element of $L_G^{(j,k)}$. Since the group generated by the union of normal subgroups is again a normal subgroup, we conclude that $L_G$ is a normal subgroup of $U_G$. Let us now show that any symmetry is, up to leaf symmetries, of the form as given in Theorem \ref{thm leaf} as stated in the following corollary.

\begin{corollary}
Let $\ket{G}\in T$ be a graph state on $n$ qubits with $U\in U_G$ and $U\not \in S_G$. Then for every element $W \in U_G/L_G$ there exists an element $V\in W$ with $V_j=\sigma_1^j\exp(i\alpha_j\sigma_2^j)$ such that $\alpha_j=m_j\pi/2^{n_j}$ with $m_j,n_j\in\mathbb{N}$.
\label{cor2}
\end{corollary}

\begin{proof}
Choose an element $A\in W$. Suppose that $A_1=\sigma_1^1\exp(i\alpha_1\sigma_2^1)$ with $\alpha_1\neq m_1\pi/2^{n_1}$ for $m_1,n_1\in\mathbb{N}$ \footnote{As mentioned before, we consider the graph states up to permutations.}. Then, by Theorem \ref{thm leaf} we know that vertex $1$ is a leaf up to local complementation. Using the same arguments as in the proof of Theorem \ref{thm leaf} we can assume wlog that $\sigma_2^1=Z$ \footnote{In contrast to the proof of Theorem \ref{thm leaf} one would need to apply the inverse LC operation afterwards as we consider here one particular graph state, not only a state up to local complementation. However, this does not alter the subsequent argument.}. Let $M_1=\{j\in N_1|N_j=\{1\}\}$ be the set of all leaves connected to $1$. For every leaf $j\in M_1$ the corresponding graph state has the symmetry $\exp(i\alpha Z_1)\otimes \exp(-i\alpha X_j)$, $\alpha\in \mathbb{R}$ (Observation \ref{char leaf}). Combining this with the fact that any symmetry of a graph state has to be of the form stated in Theorem \ref{generalform} and the group properties of $U_G$ we conclude that for all $A_j$ with $j\in M_1$ we have that $\sigma_2^j=X$ unless $\alpha_j=\pi/2+k\pi$, $k\in \mathbb{Z}$. Let $M_1'$ denote the set of $j\in M_1$ for which $\alpha_j\neq\pi/2+k\pi$. Then for every pair $(1,j)$ where $j\in M_1'$ multiply $A$ with the corresponding leaf symmetry $\exp(i\alpha_1 Z_1)\otimes \exp(-i\alpha_j X_j)$ from the right obtaining
\bea
A'\propto R e^{i(\alpha_1+\sum_{j\in M_1'}\alpha_j)Z}\otimes \one_{j\in P_l}\otimes\ldots
\eea
where $R\in\mathcal{P}_n$ contains all Pauli operators appearing in all tensor factors and the other local tensor factors are the same as in $A$ (except for multiplication by $R$). By construction we have that $A'\in W$. Computing again the symmetry $B'$ as in the proof of Theorem \ref{thm leaf} starting from $A'$ with respect to vertex $1$ we obtain that

\begin{equation}
B'\propto e^{i 2^N(\alpha_1+\sum_{j\in P_1'}\alpha_j)Z}\otimes\one
\end{equation}

where $N\in \mathbb{N}$ is the number of times we squared $A'$ during this process. By Observation \ref{2 factors} we know that

\begin{equation}
\alpha_1+\sum_{j\in P_1'}\alpha_j=\frac{2\pi k}{2^{N}}\ \ \ k\in\mathbb{Z}
\end{equation}

has to hold. Repeating this argument for all $j\not\in N_1\cup\{1\}$ with $\alpha_j\neq m_j\pi/2^{n_j}$ we obtain a representative $V\in W$ with the desired properties.
\end{proof}

Due to Corollary \ref{cor2} and Observation \ref{char leaf} we have that any stabilizer state, which is LC equivalent to a graph state in $\nlc$, possesses a continuous symmetry group and, hence, non--unitary regular local symmetries \cite{NW17}, iff the corresponding graph state contains a leaf. Combining this with the fact that any graph state not in $\nlc$ can have only LC symmetries this statement holds true for arbitrary stabilizer states.

Theorem \ref{thm leaf} together with Corollary \ref{cor2} characterizes the form of all symmetries which are non--Clifford, again for arbitrary graph states since if $\ket{G}\not\in\nlc$ it can only have LC symmetries (see Lemma \ref{LC3oneLC3all} below). We will show now that such a symmetry can only exist in case there exists a (non--trivial) Clifford symmetry of order $4$. In the subsequent section will then derive necessary and sufficient conditions for all Clifford symmetries.

\begin{corollary}
Let $\ket{G}\in T$ be a graph state on $n$ qubits with $\mathcal{S}_G\subsetneq U_G$. Then $\ket{G}$ has a local Clifford symmetry of order 4. Moreover, for any $W\in U_G/L_G$, which does not correspond to an element of $\mathcal{S}_G$, there exists $l\in \mathbb{N}$, $S\in\mathcal{S}_G$ and $V\in W$ such that $(SV)^l\in U_G$ is an LC symmetry of order $4$ and $(SV)^l\not\in \left<\mathcal{S}_G\cup L_G\right>$.
\label{cor4}
\end{corollary}

\begin{proof}
Since $\ket{G}\in T$ and $\mathcal{S}_G\subsetneq U_G$ we know from Theorem \ref{thm leaf} that $G$ either contains a leaf (up to local complementation) or there exists a $U\in U_G$, $U\not \in\mathcal{S}_G$ such that $U_j\propto \sigma_1^j\exp(i\alpha_j\sigma_2^j)$ with $\alpha_j=m_j\pi/2^{n_j}$, $m_j,n_j\in\N$ for all $j\in\{1,\ldots,n\}$. Let us first consider the case where $G$ does not contain a leaf and enumerate the qubits wlog such that $n_1\ge n_2\ge\ldots\ge n_n$. Since $U\not \in \mathcal{S}_G$ we know that $n_1\ge 2$. If $\sigma_1^1\not\propto\one$ we multiply $U$ with a suitable $S\in \mathcal{S}_G$ such that $(SU)_1\propto \exp(i\alpha_1\sigma_2^1)$. By squaring the resulting symmetry $n_1-2$ times we obtain a new symmetry $U'=(SU)^{l}$ with $l=2^{n_1-2}$. By construction $U'\in \mathcal{C}_n$, $(U')^2\not\propto\one$ and $(U')^4=\one$. Hence, $U'$ is an LC symmetry of order $4$ for $\ket{G}$.

If $G$ contains one or more leafs (up to local complementation), then for every leaf $l$ with parent $p$ we know by Observation \ref{char leaf} that in particular $\exp(i \pi/4 X_l)\otimes\exp(-i \pi/4 Z_p)$ is a symmetry of $\ket{G}$. Thus, $\ket{G}$ has an LC symmetry of order $4$. 

In order to prove the last statement in Corollary \ref{cor4} let us suppose there exists a $W\in U_G/L_G$ which does not correspond to an element of $\mathcal{S}_G$, i.e. $\ket{G}$ has more nontrivial symmetries than what is generated by leaf symmetries and its stabilizer. We consider again the representative $V$ from the proof to Corollary \ref{cor2}. Applying the same reasoning to $V$ as above to the symmetry $U$, we obtain a symmetry $V'$ which again is an LC of order $4$ for $\ket{G}$. Considering the construction of $V$ (and $V'$) it is easy to see that $V'\not\in \left<\mathcal{S}_G\cup L_G\right>$.
\end{proof}

Corollary \ref{cor4} implies that any local, non Clifford symmetry for a graph state $\ket{G}\in T$ up to multiplication with leaf symmetries and an element of the stabilizer is a root of an LC symmetry of order $4$. Since graph states not in $T$ only allow for LC symmetries (Lemma \ref{LC3oneLC3all}) we conclude that this statement holds for all graph states.

\subsection{Clifford symmetries\label{clifford sym}}
In this section we first show that it is reasonable to separate graph states in $\nlc$ from those which are not in $\nlc$. The reason for that is that the latter only admit LC symmetries. In combination with Corollary \ref{cor4}, we have that in order to identify all graph states with additional symmetries it is sufficient to characterize those with LC symmetries. We then present necessary and sufficient conditions on the adjacency matrix of the corresponding graph to identify these symmetries. 

Recall that any graph state not in $T$ has a symmetry $U\in U_G$ such that $U_j\in \mathcal{C}_1^3$ for at least one $j\in\{1,\ldots,n\}$. In the following lemma we show that such graph states can only have LC symmetries. 

\begin{lemma}
Let $\ket{G}$ be a graph state on $n$ qubits and let $U\in U_G$ be such that $U_1\in\mathcal{C}_1^3$. Then $U_j\in \mathcal{C}_1^3$ for all $j\in\{1,\ldots,n\}$. Moreover, any other symmetry of the graph state is a local Clifford operator, i.e. $U_G\subseteq \mathcal{C}_n$.
\label{LC3oneLC3all}
\end{lemma}

A proof for this lemma is provided in Appendix \ref{AB}. This lemma implies that any graph state not in $T$ only allows for LC symmetries. Furthermore, there is no mixing between LC factors of order $3$ and order $4$ within a single symmetry. Note however that there are graph states for which an LC symmetry of order $3$ is a product of two different LC symmetries of order $4$ (graph $a)$ in Fig. \ref{graphs}).

According to Corollary \ref{cor4} any graph state in $T$ with an additional symmetry ($\mathcal{S}_G\subsetneq U_G$) also has an LC symmetry of order $4$. Combining this with Lemma \ref{LC3oneLC3all} we see that in order to find all graphs with additional symmetries we have to identify which graph states admit LC symmetries of order $3$ and $4$. Let us now derive necessary and sufficient conditions for the existence of these symmetries. As we will see, they lead to different conditions on the adjacency matrix of the graph depending on whether the graph state possesses symmetries of order 3 or 4. In \cite{VN04} the following theorem, which is crucial for the characterization of LC symmetries, has been shown.

\begin{theorem}[\cite{VN04}]\label{Th:LocalClifford} Two graph states $\ket{G}$ and $\ket{G'}$ on $n$ qubits defined by adjacency matrices $\theta$ and $\theta'$ are related via a local Clifford operation iff there exist diagonal binary matrices $A,B,C,D\in M_{n\times n}$ satisfying

\begin{equation}
AD+BC=\one
\label{eq:conditionclifford}
\end{equation}

such that

\begin{equation}
0=\theta' C\theta+ A\theta+\theta' D+B.
\label{consforclifford}
\end{equation}
\end{theorem}

All computations are carried out over $\mathbb{Z}_2$. To show this theorem the authors make use of the fact that the stabilizer formalism has a representation in terms of binary matrices 
 \cite{NM04, GoDa97}. In this representation a Pauli operator $p\in \mathcal{P}_1$ corresponds to the following $2\times 1$ matrices, which we denote by $b(p)$,
\bea
\begin{split}\ b(\one)= \begin{pmatrix}0\\0\end{pmatrix},\ b(X)=\begin{pmatrix}0\\1\end{pmatrix},\\ b(Y)= \begin{pmatrix}1\\1\end{pmatrix},\ b(Z)= \begin{pmatrix}1\\0\end{pmatrix}.
\end{split}
\eea
Analogously an element $p\in\mathcal{P}_n$ is represented by a $2n\times 1$ matrix $b(p)$ which is defined as follows

\bea
b(p)_{i,1}&=b(p_i)_{1,1}\\
b(p)_{i+n,1}&=b(p_i)_{2,1},
\eea
for $i\in \{1,\ldots,n\}$. Note that this representation does not contain information about additional phases. Furthermore, in this representation a graph state can be associated to a $2n\times 2n$ matrix where the columns represent a set of generators for its stabilizer. Using the canonical set of generators this matrix is $(\theta,\one)$ where $\theta$ is the adjacency matrix of the corresponding graph.

In this representation the action of a local Clifford operator on a graph state corresponds to the multiplication of $(\theta,\one)$ (from the left) with a $2n\times 2n$ matrix 
\begin{equation}
Q=\begin{pmatrix}
	A&B\\
	C&D
\end{pmatrix},
\label{eq:fromclifford}
\end{equation}
where $A,B,C,D$ are $n\times n$ diagonal matrices. The Clifford operation applied to qubit $j$ is given by the submatrix $Q_j=((A_{jj},B_{jj}),(C_{jj},D_{jj}))$ and, as it is invertible, it has to satisfy $\text{det } Q_j=1$. As mentioned before, phases are not represented in this picture, that is e.g. $X$ and $ZXZ=-X$ have the same representation. This implies that Clifford operators which are related to each other via local Pauli operators are mapped to the same matrix $Q$. As $|\mathcal{C}_1\backslash \mathcal{P}_1|=6$, a single qubit Clifford operator is mapped to one out of $6$ different matrices $Q$ by this representation. Table \ref{tablebinaryreal} shows the matrix $Q$ together with a representative of the corresponding equivalence class in $\mathcal{C}_1\backslash \mathcal{P}_1$.
\begin{table}[h]

\centering

\begin{tabular}{|c||c|c|c|c|c|c|}

\hline
\renewcommand{\arraystretch}{0.5}
\parbox[c][0.9 cm][c]{0.04\textwidth}{$b(C)$}&

$\begin{pmatrix}
	1&0\\
	0&1
\end{pmatrix}$&

$\begin{pmatrix}
	1&1\\
	0&1
\end{pmatrix}$&

$\begin{pmatrix}
	0&1\\
	1&0
\end{pmatrix}$&

$\begin{pmatrix}
	1&0\\
	1&1
\end{pmatrix}$&

$\begin{pmatrix}
	0&1\\
	1&1
\end{pmatrix}$&

$\begin{pmatrix}
	1&1\\
	1&0
\end{pmatrix}$\\
\hline
\parbox[c][0.6 cm][c]{0.04\textwidth}{$C$} &

$\one$&

$e^{i\frac{\pi}{4}Z}$&

$e^{i\frac{\pi}{4}Y}$&

$e^{i\frac{\pi}{4}X}$&

$e^{i\frac{\pi}{4}Z}e^{i\frac{\pi}{4}Y}$&

$e^{i\frac{\pi}{4}X}e^{i\frac{\pi}{4}Y}$\\\hline

\end{tabular}
\caption{Local Clifford operations of order $3$ (last two) and $4$ in the binary and standard representation. \label{tablebinaryreal}}
\end{table}

We utilize now Theorem \ref{Th:LocalClifford} together with the binary representation explained above to determine the LC symmetries of an arbitrary graph state. For our purpose we consider $\theta=\theta'$. As we will see, solving Eq. (\ref{consforclifford}) leads to two possible cases. The first characterizes all LC symmetries of order 3 and the second characterizes LC symmetries of order 4, as stated in the subsequent theorems.

As explained above, the binary representation of the stabilizer does not allow to determine local Pauli operators (as they only change the sign of elements of the stabilizer under conjugation). However, note that for any $U$ with $U \ket{G}\propto P \ket{G}$ it holds that $ Z^{\vec {k}} U \ket{G} \propto  \ket{G}$, where $\vec{k}$ is such that $P \ket{G}\propto Z^{\vec {k}} \ket{G}$ \cite{HeMa06}. Thus, to determine the local symmetry including local Pauli operators we choose a representative of $U$ found as explained below and check whether $Z^{\vec{k}}U$ is a symmetry for some $\vec{k}$. Note that for any $U$ there exists a unique $\vec{k}$ such that $Z^{\vec{k}}U\in U_G$. Otherwise there would exist an element of $\mathcal{S}_G$ that is just a tensor product of $Z$ operators and $\one$ (which is not possible as can be easily verified considering the canonical generators). Let us now state the theorem which identifies graph states with LC symmetries of order $3$.
\begin{theorem}
Let $\ket{G}$ be a graph state on $n$ qubits and let $\theta$ be the adjacency matrix of the corresponding graph. Then, there exists some $U\in U_G$ with $U_j\in\mathcal{C}_1^3$ (for some $j\in\{1,\ldots,n\}$) iff there exists $d\le n$ and an ordering of the vertices such that
\bea 
\theta^2=\begin{pmatrix}
           \theta_{00}+\one & 0 \\
           0 & \theta_{11}+\one
         \end{pmatrix},\label{eq:lc31}
\eea
where $\theta=((\theta_{00},\theta_{01}),(\theta_{10},\theta_{11}))$, $\theta_{00}\in M_{d\times d}$, and $\theta_{11}\in M_{(n-d)\times (n-d)}$. Furthermore, the solutions to Eq. (\ref{eq:lc31}) correspond uniquely to LC symmetries of order $3$ of $\ket{G}$ (up to multiplication (from left and right) by elements of the stabilizer). More precisely, given an ordering for which a solution exists, the symmetry is given by $V^{\otimes d} \otimes W^{\otimes n-d}$, where $V=e^{\pm i \pi/4 X}e^{\pm i \pi/4 Y}$ and $W=e^{\pm i \pi/4 Z}e^{\pm i \pi/4 Y}$, for some choice of signs of the phases (independently on each qubit). Moreover, it holds that $(-1)^{k_1}Y^{\otimes n},(-1)^{k_2}X^{\otimes d}\otimes Z^{\otimes n-d}\in\mathcal{S}_G$ for some $k_1,k_2\in \{0,1\}$ and that $d\neq 0,n$.
\label{conditionlc3}
\end{theorem}

According to Theorem \ref{conditionlc3} any LC symmetry of order $3$ of a graph state has to be a product of square roots of two elements of the stabilizer (up to local Pauli operators), which up to permutations of the qubits are of the form $Y^{\otimes n},X^{\otimes d}\otimes Z^{\otimes n-d}$. Note that, as mentioned before, using the binary representation of the stabilizer state, the symmetries can only be determined up to local Pauli operators. However, as stated in the theorem, these Pauli operators do not need to be computed, as it can be shown that their effect can be compensated by choosing the phases properly. 
More precisely, it can be shown that if $P V^{\otimes d} \otimes W^{\otimes n-d} \in U_G$, then, there exists an assignment of the phases in the exponent of the LC operations of order $3$, such that the resulting LC operator of order $3$ is a symmetry with no additional Pauli operator (see proof of Theorem \ref{conditionlc3}). Hence, the local Pauli operators are not required in this case. Before proving this theorem, let us state the necessary and sufficient conditions for the existence of a LC symmetry of order $4$.

\begin{theorem}
Let $\ket{G}$ be a graph state on $n$ qubits and let $\theta$ be the adjacency matrix of the corresponding graph. Then, there exists $U\in U_G$ such that $U\in\mathcal{C}_n$, $U\not\in\mathcal{S}_G$ and $U^4=\one$ iff there exists $d\le n$ and an ordering of the vertices such that

\bea 
(\theta_{00}+X)^2=0\label{eq:lc41}\\
(\theta_{00}+X) \theta_{01}=0\\
\theta_{01}^T\theta_{01}+Y&=0\label{eq:lc43},
\eea

where $X$ and $Y$ are diagonal matrices with $X_{ii}=\sum_{j\le d}\theta_{ij}$, $Y_{ii}=\sum_{j\le d}\theta_{i+d,j}$, $\theta=((\theta_{00},\theta_{01}),(\theta_{10},\theta_{11}))$, $\theta_{00}\in M_{d\times d}$ and $\theta_{11}\in M_{(n-d)\times (n-d)}$.  Furthermore, the solutions to Eq. (\ref{eq:lc41}) to (\ref{eq:lc43}) correspond uniquely to a symmetry $U$ (up to local Pauli operators). Moreover, if $\bigotimes_j e^{i\pi/4 \sigma_j}\otimes \one$ is an LC symmetry of $\ket{G}$ (up to local Pauli operators) then $(-1)^k\bigotimes_j \sigma_j\in\mathcal{S}_G$ for some $k\in\{0,1\}$.
\label{conditionlc4}
\end{theorem}

We use Theorem \ref{Th:LocalClifford} to prove Theorem \ref{conditionlc3} and \ref{conditionlc4} \footnote{Note that for the symmetries in Theorem \ref{conditionlc3} we require that only one of the tensor factors is a Clifford of order $3$. However, Lemma \ref{LC3oneLC3all} implies that then all the other factors are in $\mathcal{C}_1^3$ as well.}. As we will see, we have to consider two cases, where one corresponds to the proof of Theorem \ref{conditionlc3} and the other to the proof of Theorem \ref{conditionlc4}. Using $\theta=\theta'$ in Eq. (\ref{consforclifford}) leads to
\begin{equation}
0=\theta C\theta+ A\theta+\theta D+B.
\label{eq:conditioncliffordsymmetry}
\end{equation}
As $\theta$ is symmetric we also have
\begin{equation}
0=\theta C\theta+ \theta A+D\theta +B.
\end{equation}
Adding both equations (modulo 2) leads to
\begin{equation}
[\theta,A+D]=0.
\label{commute}
\end{equation}
As $\theta$ corresponds to a connected graph, i.e. in each column and each row there exists at least one non--vanishing entry, the last equation is fulfilled iff $A+D=0,\one$. We treat case (i), where $A=D+\one$, which corresponds to LC of order 3 (Theorem \ref{conditionlc3}) and case (ii), where $A=D$, which corresponds to LC of order 4 (Theorem \ref{conditionlc4}), separately (see also Table \ref{tablebinaryreal}).

\begin{proof}[Proof to Theorem \ref{conditionlc3}]
Let us first consider case (i), i.e. $A=D+\one$. Inserting this in Eq. (\ref{eq:conditionclifford}) leads to $BC=\one$ or equivalently to $B=C=\one$. Hence, the Clifford operator (of order 3) applied to $\ket{G}$ is of the form
	\begin{equation}
	Q=\begin{pmatrix}
		A&\one\\
		\one& A+\one
	\end{pmatrix}.
	\end{equation}
As can be easily seen (see Table \ref{tablebinaryreal}), this implies that there exists an LC of order 3 iff all parts of the symmetry are LCs of order 3. This provides an alternative proof to the fact that LC symmetries, which contain factors of order $3$, have a Clifford operator of order $3$ on every qubit (see Lemma \ref{LC3oneLC3all}). Using these findings in Eq. (\ref{eq:conditioncliffordsymmetry}) we get
	\begin{equation}
	0=\theta^2+A\theta+\theta A+\theta+\one\label{locl1}.
	\end{equation}
Choosing now an ordering of the vertices in the graph such that $A=((\one_{d\times d},0),(0,0))$ we find
	\begin{equation}
	\begin{split}
	0=&\begin{pmatrix}
		\theta_{00}&\theta_{01}\\
		\theta_{01}^T&\theta_{11}
	\end{pmatrix}^2+\begin{pmatrix}
		\theta_{00}&\theta_{01}\\
		0&0
	\end{pmatrix}\\&+\begin{pmatrix}
		\theta_{00}&0\\
		\theta_{01}^T&0
	\end{pmatrix}+\begin{pmatrix}
		\theta_{00}&\theta_{01}\\
		\theta_{01}^T&\theta_{11}
	\end{pmatrix}+\one.\label{locl5}\end{split}
  \end{equation}
As the computation is modulo 2, we obtain the necessary and sufficient condition stated in Theorem \ref{conditionlc3}. Up to local Pauli operators the operator $V^{\otimes d} \otimes W^{\otimes n-d}$ with $V=e^{i \pi/4 X}e^{i \pi/4 Y}$ and $W=e^{i \pi/4 Z}e^{i \pi/4 Y}$ is the operator corresponding to $Q$ (see Table \ref{tablebinaryreal}). As mentioned above, for any $U$ such that $U\ket{G}\propto P\ket{G}$ the Pauli operator $P$ is unique up to multiplication with elements of the stabilizer. Thus, we conclude that $Q$ uniquely corresponds to an LC symmetry of order $3$ up to multiplication with elements of the stabilizer. Next we show that there exists a choice for the signs in the exponent of the symmetry such that $P$ can be chosen to be the identity. Let $P(V^{\otimes d} \otimes W^{\otimes n-d})\in U_G$. Every local factor of this symmetry is of the form $P_j\exp(i(-1)^{k_1^j}\pi/4 \sigma_1^j)\exp(i(-1)^{k_j^2}\pi/4 \sigma_2^j)$, $k_1^j=k_2^j=0$ for all $j$. Using that $\exp(i\pi/4 \sigma)=i\sigma\exp(-i\pi/4 \sigma)$ for $\sigma\in \{X,Y,Z\}$ it is clear that by choosing appropriate values for the variables $k_1^j$ and $k_2^j$ one can obtain the additional factor $P_j$. Using this for every qubit the claim follows.

Next we show that $Y^{\otimes n},X^{\otimes d}\otimes Z^{\otimes n-d}\in\mathcal{S}_G$. Observe that the only operators in Eq. (\ref{locl1}) with nonzero entries on the diagonal are $\one$ and $\theta^2$. Thus, for every qubit $j$ it has to hold that $\sum_k\theta_{jk}\theta_{jk}=\sum_k\theta_{jk}=1$ which is equivalent to the statement that every qubit has an odd number of neighbours. Thus, the product of all canonical generators yields $\prod_j S_{(j)}\propto Y^{\otimes n}$ and consequently $(-1)^{k_1} Y^{\otimes n}\in \mathcal{S}_G$ for some $k_1\in \{0,1\}$. To show that $X^{\otimes d}\otimes Z^{\otimes n-d}\in\mathcal{S}_G$ let us write the operator corresponding to the symmetry $Q$ as $PC_1C_2$ where $P\in\mathcal{P}_n$, $C_1=(e^{i \pi/4 X})^{\otimes d}\otimes (e^{i \pi/4 Z})^{\otimes (n-d)}$ and $C_2=(e^{i \pi/4 Y})^{\otimes n}$. Here, $P$ is chosen such that the sign of the exponent of all tensor factors is positive. Since $(-1)^{k_1}Y^{\otimes n}\in \mathcal{S}_G$ and $[Y^{\otimes n},C_2]=0$ also $(-1)^{k_1}PC_1Y^{\otimes n}C_2\in U_G$ and $PC_1Y^{\otimes n}C_2(P C_1 C_2)^\dagger Y^{\otimes n}= P C_1Y^{\otimes n}C_1^\dagger P^\dagger Y^{\otimes n} =\pm C_1^2\in U_G$ where the last equation holds as $\pm C_1^2\in\mathcal{P}_n$. Thus, $\pm C_1^2\ket{G}=\ket{G}$. Due to Observation \ref{nootherpauli} we have that $C_1^2\in \mathcal{S}_G$. Using the canonical generators we find that $\pm C_1^2=(-1)^{k_2}X^{\otimes d}Z^{\otimes(n-d)}=\prod _{j\le d} S_{(j)}$ for $k_2\in \{0,1\}$ has to hold. Thus, we have shown that $(-1)^{k_2}X^{\otimes d}\otimes Z^{\otimes n-d}\in\mathcal{S}_G$.

Finally, it remains to show that we can only find solutions for $d\neq 0,n$. If $d=0$ then $(-1)^{k_2} Z^{\otimes n}\in\mathcal{S}_G$ which is not possible as the canonical generators are of the form $S_{(j)}=X_j\otimes \bigotimes_{k\in N_j}Z_j$. Similarly, if $d=n$ then $(-1)^{k_2}X^{\otimes n}\in\mathcal{S}_G$ and $(-1)^{k_1}Y^{\otimes n}\in\mathcal{S}_G$ and thus again $(-1)^{k_1+k_2}(i)^n Z^{\otimes n}\in\mathcal{S}_G$. Thus, we conclude that $d\neq 0,n$.
\end{proof}

Let us now consider the remaining case (case(ii)) to prove Theorem \ref{conditionlc4}.

\begin{proof}[Proof to Theorem \ref{conditionlc4}]
Using $A+D=0$ and that $A$ and $D$ are diagonal matrices in Eq. (\ref{eq:conditionclifford}) gives $A^2+BC=A+BC=\one$ and, thus, the form of the Clifford operator applied to $G$ is
	\begin{equation}
	Q=\begin{pmatrix}
		BC+\one&B\\
		C&BC+\one
	\end{pmatrix}.
	\end{equation}
From Table \ref{tablebinaryreal} we see that all these symmetries correspond to LCs of order 4 (up to multiplication by local Pauli operators). Using these findings in Eq. (\ref{eq:conditioncliffordsymmetry}) we get
	\begin{equation}
	0=\theta C\theta +BC\theta+\theta BC+B\label{locl2}.
	\end{equation}
The only two summands with non-vanishing diagonal are $B$ and $\theta C\theta$ and thus we find $B_{ii}=\sum_{j}\theta_{ij}C_{jj}$. Choosing an ordering such that $C=((\one,0),(0,0))$ we find
	\begin{align}
	\begin{split}
	0=&\begin{pmatrix}
		\theta_{00}^2&\theta_{00}\theta_{01}\\
		\theta_{01}^T\theta_{00}&\theta_{01}^T\theta_{01}
	\end{pmatrix}+ \begin{pmatrix}
		X&0\\
		0&0
	\end{pmatrix}\begin{pmatrix}
		\theta_{00}&\theta_{01}\\
		\theta_{01}^T&\theta_{11}
	\end{pmatrix}\\&+\begin{pmatrix}
		\theta_{00}&\theta_{01}\\
		\theta_{01}^T&\theta_{11}
	\end{pmatrix}\begin{pmatrix}
		X&0\\
		0&0
	\end{pmatrix}+\begin{pmatrix}
		X&0\\
		0&Y
	\end{pmatrix},\end{split}\label{locl3}
	\end{align}
where we defined $B=((X,0),(0,Y))$ according to the ordering defined by choosing $C=((\one,0),(0,0))$. It is straightforward to see that the equations above are equivalent to the ones given in Theorem \ref{conditionlc4}. As mentioned above, for any $U$ such that $U\ket{G}\propto P\ket{G}$ the Pauli operator $P$ is unique up to multiplication with elements of the stabilizer. Thus, we conclude that $Q$ uniquely corresponds to an LC symmetry of order $4$ up to multiplication with elements of the stabilizer.

It remains to show that if $\bigotimes_j e^{i\pi/4 \sigma_j}\otimes \one$ is the symmetry corresponding to $Q$ (up to local Pauli operators) then $(-1)^k\bigotimes_j \sigma_j\in\mathcal{S}_G$ for some $k\in \{0,1\}$. To see this consider an ordering of the vertices such that $C=\text{diag}(\one,\one,0,0)$ and $B=\text{diag}(0,\one,\one,0)$. We denote by $D_1$, $D_2$, $D_3$ and $D_4$ the set of qubits corresponding to the respective blocks. Note that some of these sets may be empty. Furthermore let us denote the corresponding blocks of $\theta$ by $t_{jk}$ where $j,k\in\{0,1,2,3\}$ and note that $t_{jk}^T=t_{kj}$ as $\theta$ is symmetric. Using this block structure in Eq. (\ref{locl2}), the equations on the diagonal read
\bea
t_{00}^2+t_{01}t_{10}&=&0\label{eee1}\\
\one+t_{10}t_{01}+t_{11}^2&=&0\label{eee2}\\
\one+t_{20}t_{02}+t_{21}t_{12}&=&0\label{eee3}\\
t_{30}t_{03}+t_{31}t_{13}&=&0\label{eee4}.
\eea
Note that for any qubit $j\in D_1$, the corresponding operator in the symmetry is $\sigma_j=X$, for $j\in D_2$ it is $\sigma_j=Y$, for $j\in D_3$ it is $\sigma_j=Z$ and for $j\in D_4$ the symmetry acts trivial.

Let us consider Eq. (\ref{eee1}). The matrix element $(t_{00}^2)_{jj}$ is nonzero if qubit $j$ has an odd number of neighbours in $D_1$. The term $t_{01}t_{10}$ has a nonzero entry on the diagonal if the corresponding qubit $j$ has an odd number of neighbours in $D_2$. Thus, for Eq. (\ref{eee1}) to hold every qubit $j$ with $\sigma_j=X$ has to have an even number of neighbours in $D_1$ and $D_2$ (summed up). Analyzing the other equations in the same way we find that Eq. (\ref{eee2}) implies that every qubit $j$ with $\sigma_j=Y$ has to have an odd number of neighbours in $D_1$ and $D_2$, Eq. (\ref{eee3}) implies that every qubit $j$ with $\sigma_j=Z$ has to have an odd number of neighbours in $D_1$ and $D_2$ and Eq. (\ref{eee4}) implies that every qubit $j$ on which the symmetry acts trivial has to have an even number of neighbours in $D_1$ and $D_2$. Let us now consider the product of the canonical generators  $\prod_{j\in D_1\cup D_2}S_{(j)}$ corresponding to the qubits in $D_1$ and $D_2$. As $S_{(j)}=X_j\bigotimes_{k\in N_j}Z_k$ the number of $Z$ operators acting on qubit $j$ in this product is determined by the number of neighbours of qubit $j$ in $D_1$ and $D_2$. Combining this with the considerations from above we conclude that $\prod_{j\in D_1\cup D_2}S_{(j)}=(-1)^k\bigotimes_j\sigma_j\otimes \one$ for some $k\in\{0,1\}$ and thus $(-1)^k\bigotimes_j\sigma_j\otimes \one\in\mathcal{S}_G$. 
\end{proof}

Let us summarize the results of the previous sections (see also Figure \ref{sum proof}). Any local unitary symmetry of a graph state $\ket{G}$ has to be of the form specified in Theorem \ref{generalform}. Combining this with Lemma \ref{LC3oneLC3all} we find that there are two types of graph states, those which do not have LC symmetries of order $3$, forming the set $T$ (Eq. (\ref{setT})), and those which do have LC symmetries of order $3$ ($\neg T$). 

The set $T$ contains all graph states with no additonal symmetries, i.e. with $U_G=\mathcal{S}_G$. Furthermore, graph states in $T$ can have additional continuous symmetries iff the corresponding graph has a leaf up to local complementation (Theorem \ref{thm leaf}, Observation \ref{char leaf}, Corollary \ref{cor2}). Note that this continuous symmetry also includes some Clifford symmetries. Up to multiplication by leaf symmetries any other local symmetry of a graph state in $T$ which is not an element of the stabilizer is (up to multiplication by elements of $\mathcal{S}_G$) a $2^k$th root of an LC symmetry of order $4$ of the state up to local Pauli operators (Theorem \ref{thm leaf}, Corollary \ref{cor2} and Corollary \ref{cor4}). Corollary \ref{bounds} shows that there can be only finitely many of these additional symmetries for any graph state. Their form is characterized by Theorem \ref{thm leaf}.

Any graph state in $\neg T$ has an LC symmetry of order $3$. According to Lemma \ref{LC3oneLC3all} any graph state with this property can only have LC symmetries. Combing this with the fact that for any symmetry of a graph state in $T$ there exists an LC symmetry of order $4$ of the corresponding graph (Corollary \ref{cor4}) we conclude that in order to find all graph states with additional symmetries we have to find those which admit nontrivial LC symmetries. In Theorem \ref{conditionlc3} and \ref{conditionlc4} we present sets of equations for the adjacency matrix of a graph to find all LC symmetries of order $3$ and $4$ (up to local Pauli operators) of the corresponding graph state. Let us remark that it is possible that for some specific graph states an LC symmetry of order $3$ is a product of two LC symmetries of order $4$ of the same graph. Furthermore, Theorem \ref{conditionlc3} and \ref{conditionlc4} show that in order to find all LC symmetries of a graph state it is sufficient to check square roots of elements of the stabilizer up to signs (or products of them in the case of LC symmetries of order $3$). In case of LC symmetries of order $3$ there exists a choice of the signs in the exponent such that no additional Pauli operators are needed. In case of LC symmetries of order $4$ additional Pauli operators have to be taken into account.

The results presented here lead to an algorithm to find all local (unitary) symmetries of a graph state. We present this algorithm in Sec. \ref{alg all sym}.

\tikzstyle{central} = [ draw,fill=blue!20, 
    text width=17em, text badly centered, node distance=3cm, inner sep=5pt]
    \tikzstyle{central1} = [ draw,fill=blue!20, 
    text width=4.5em, text badly centered, node distance=3cm, inner sep=5pt, minimum height=3.7cm]
       \tikzstyle{blank} = [ fill=blue!00, 
    text width=6em, text badly centered, node distance=3cm, inner sep=0pt]
        \tikzstyle{set1} = [ draw,fill=red!30, 
    text width=4.5em, text badly centered, node distance=3cm, inner sep=5pt,rounded corners]
    \tikzstyle{set} = [rectangle,draw, fill=red!30, 
    text width=8em, text badly centered,rounded corners, node distance=3cm, inner sep=5pt]
    \tikzstyle{dummy} = [ fill=red!00, 
    text width=0.1em, text badly centered, node distance=3cm, inner sep=1pt]
\tikzstyle{line} = [draw, -latex']

\pgfdeclarelayer{bg}
\pgfsetlayers{bg,main}

\begin{figure}[h!]
\centering
\begin{tikzpicture}[node distance = 2cm, auto]

\node[central] at (0,-1.75)(1){\scriptsize general form of symmetries (Theorem \ref{generalform}) $U_j\propto \left\{\begin{matrix}C_j \\\sigma_1^j\exp\left(i\alpha_j\sigma_2^j\right)
\end{matrix}\right.$ };

\node[set]  at (0,0)(2){\scriptsize$\{\ket{G}\}$};

\node[set] at (2,-3.5)(3){\scriptsize$\ket{G}\in\neg T$ (Eq. \ref{setT})};

\node[set] at (-2,-3.5)(4){\scriptsize$\ket{G}\in T$ (Eq. \ref{setT})};

\node[central1] at (-3,-6)(5){\scriptsize discrete symmetries $\alpha_j=\frac{m_j\pi}{2^{n_j}}$ (Theorem \ref{thm leaf}, Corollary \ref{cor2})};

\node[central1]at (-1,-6)(6){\scriptsize continuous symmetry: leaf (Corollary \ref{cor2}, Observation \ref{char leaf}) $\rightarrow$ $GL$ symmetries};

\node[set1]at (1,-5)(7){\scriptsize $\ket{G}$ with $U_G\subseteq \mathcal{C}_n$ (Lemma \ref{LC3oneLC3all})};

\node[set1]at (3,-5)(8){\scriptsize $\ket{G}$ with $U_G\not \subseteq \mathcal{C}_n$ excluded (Lemma \ref{LC3oneLC3all})};

\node[central]at (0,-10)(9){\scriptsize LC symmetries (Theorem \ref{conditionlc3} and Theorem \ref{conditionlc4})};

\node[blank]at (-2,-9)(10){\scriptsize Corollary \ref{cor4}};

\begin{pgfonlayer}{bg}
         \path [line] (2)  -- (3);
        \path [line] (2)  -- (4);
        \path [line] (3)  -| (7);
        \path [line] (3)  -| (8);
        \path [line, dashed] (2)  -- (1);
        \path [line, dashed] (4)  -- (5);
        \path [line, dashed] (4)  -- (6);
        \path [line, dashed] (7)  -- (9);
        \path [line, dashed] (5)  -- (10);
        \path [line, dashed] (6)  -- (10);
        \path [line, dashed] (10)  -- (9);
    \end{pgfonlayer}
\end{tikzpicture}
\caption{Summary of the main steps of the characterization of the local invertible symmetry group $U_G$ for arbitrary graph states. Red boxes denote sets of graph states and blue boxes contain results on the local symmetries of these states\label{sum proof}.}
\end{figure}
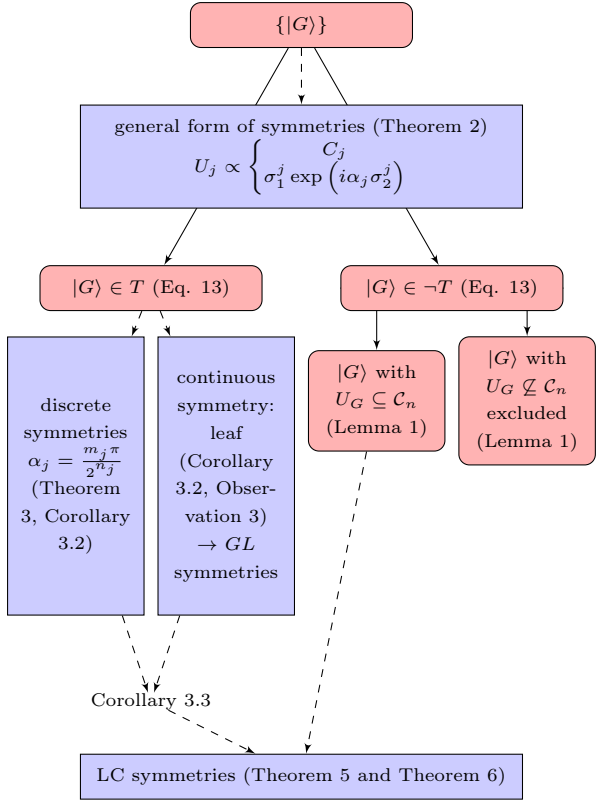

\subsection{Non--unitary symmetries\label{non-unitary symmetries}}

First note that there only exist symmetries in $GL$ iff there is a (at least) 1--parameter family of local unitary symmetries, as shown in \cite{NW17}. As shown above, this is only the case if the graph contains a leaf up to local complementation. In case vertex $l$ is a leaf and vertex $p$ is its parent the local symmetries in $GL$ are given by $\exp(i\alpha Z_p) \otimes \exp(-i\alpha X_l)\otimes \one$ with $\alpha\in \mathbb{C}$. The symmetries for twin vertices and connected twins follow analogously.

Let us stress here that not only invertible local symmetries play a role in the study of separable maps transforming one state into the other. In fact, one also needs to consider local projectors, which annihilate the initial state \cite{HeMa19, GoGi19cor}. In the following we present a general recipe to construct some of these projectors for stabilizer states. 

\begin{lemma} 
Let $ \ket{\psi}$ be a stabilizer state and let $S\in \mathcal{S}_G$ with $S=S_1\otimes S_2 \otimes \ldots \otimes S_k\otimes \one$ and $S_j\not\propto\one$ for all $j\in\{1,\ldots k\}$ be an element of the stabilizer. Then, for\bea
Q_{S}^f=\bigotimes_{j=1}^k(\one+(-1)^{f(j)}S_j)
\eea
it holds that 
\bea
Q_{S}^f\ket{G}=0
\eea
for all $f:\{1,\ldots,k\}\rightarrow \{0,1\}$ such that $\sum_{j=1}^k f(j)$ is odd.
\label{annihilate}
\end{lemma}

\begin{proof}
Let $f$ be an arbitrary function $f:\{1,\ldots, k\}\rightarrow \{0,1\}$ with $\sum_{j=1}^kf(j)$ odd. First observe that $\bigotimes_{j=1}^k(\one+(-1)^{f(j)}S_j)\propto U\ket{f(1),\ldots,f(k)}\bra{f(1),\ldots,f(k)}U^\dagger$ where $U=\bigotimes _j U_j$ is such that $S_j=U_j(\ket{0}\bra{0}-\ket{1}\bra{1})U_j^\dagger$ for all $j\in\{1,\ldots,k\}$. To show the claim it is sufficient to show that $U\ket{f(1),\ldots,f(n)}$ for $f$ chosen as above is in the kernel of $\rho=\ket{G}\bra{G}$. We use that $\rho\propto\prod_{i=1}^n(\one+g_i)$ where $\{g_i\}_{i=1}^n$ is a set of generators for $\mathcal{S}_G$. We choose this set such that $g_1=S$. As the operators $\one+g_i$ commute, we find that the kernel of $\rho$ is given by $\text{span}\{\text{ker}(\one + g_i)\}_i$. Observe that $\text{ker}(\one + g_1)=\text{span}[\{U\ket{f(1),\ldots,f(k)}|f:\{1,\ldots,k\}\rightarrow \{0,1\},\sum_{j=1}^kf(j)\ \text{odd}\}]$. This shows the claim.
\end{proof} 

\subsection{Algorithm to determine all symmetries of a graph state\label{alg all sym}}
We present here an algorithm to find all local (unitary) symmetries of a graph state (see also Figure \ref{algorithm}). Recall that the only non unitary, invertible symmetries stem from a complexification of leaf symmetries.

Let $\ket{G}$ be a graph state and let $\theta$ be the corresponding adjacency matrix. The first step of the algorithm is to determine all LC symmetries of $\ket{G}$ (step $(1)$ in Figure \ref{algorithm}). Use Theorem \ref{conditionlc3} and \ref{conditionlc4} to determine the symmetries up to local Pauli operators. Recall that for any $P\in\mathcal{P}_n$ we have that $P\ket{G}\propto Z^{\vec{k}}\ket{G}$. Thus, to determine the exact expression for an LC symmetry $U_{(j)}$, i.e. to determine the additional Pauli operators, find the vector $\vec{k}_{(j)}$ such that $Z^{\vec{k}_{(j)}}U_{(j)}\ket{G}\propto \ket{G}$ by going through all possibilities. All LC symmetries of the graph are then given by the group $U_{LC}=\left<\{Z^{\vec{k}_{(j)}}U_{(j)}\}\cup \mathcal{S}_G\right>$.

If $U_{LC}$ contains an LC of order $3$ then we know by Lemma \ref{LC3oneLC3all} that $\ket{G}$ only has LC symmetries and thus $U_G=U_{LC}$ (step $(2)$ in Figure \ref{algorithm}). In case the graph state does not possess an LC symmetry of order $3$, i.e. $\ket{G}\in T$, we again distinguish two cases (step $(3)$ in Figure \ref{algorithm}). If $\ket{G}$ does not have an LC symmetry of order $4$ then Corollary \ref{cor4} implies that $U_G=\mathcal{S}_G$ and the state does not have any additional symmetries. In case the state has LC symmetries of order $4$, check if all of these LC symmetries are generated by leaf symmetries (step $(4)$ in Algorithm \ref{algorithm}). To do so, find all leafs, parents, twins and connected twins of the graph using the adjacency matrix $\theta$ \footnote{Equivalently, one could compute all two qubit reduced states $\rho_{ij}$ for $\left|G\right >$. It holds that $\rho_{ij}\not\propto \one$ iff $i$ and $j$ are leaf and parent up to local complementation.}. If all LC symmetries of the graph stem from leaf symmetries then the graph has no other additional symmetries except for the leaf symmetries (Corollary \ref{cor4}).

\tikzstyle{decision} = [diamond, draw, fill=red!30, 
    text width=4.5em, text badly centered, node distance=3cm, inner sep=0pt]
\tikzstyle{block} = [rectangle, draw, fill=blue!30, 
    text width=5em, text centered, rounded corners, minimum height=4em]
   \tikzstyle{title} = [rectangle, draw, fill=green!30, 
    text width=20em, text centered, rounded corners, minimum height=4em]
\tikzstyle{line} = [draw, -latex']

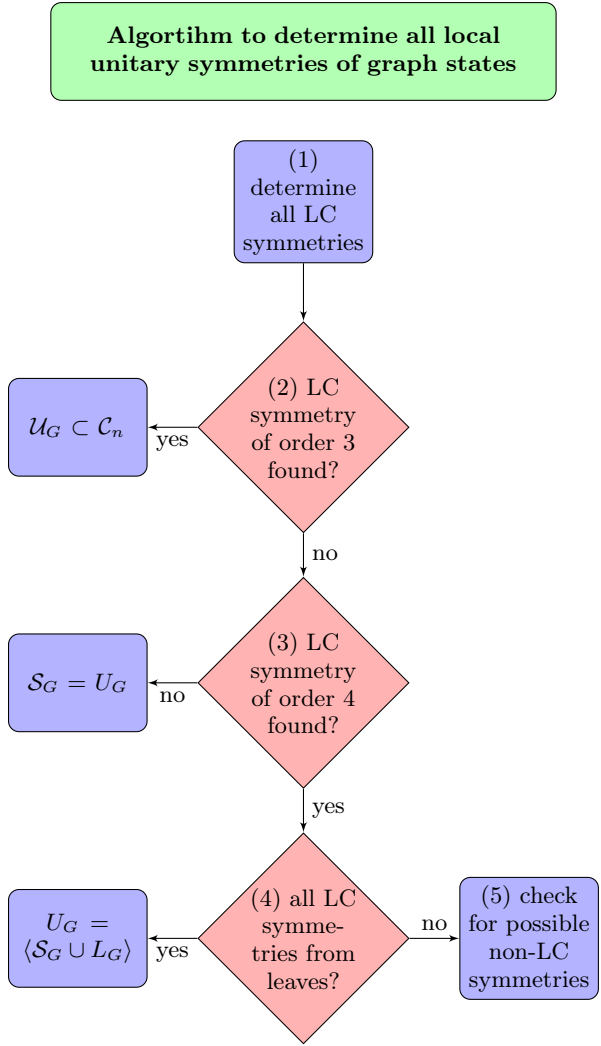
\begin{figure}[h!]
\centering
\begin{tikzpicture}[node distance = 2cm, auto]
	\node[title] (ti){\bf Algortihm to determine all local unitary symmetries of graph states};
    \node [block, below of=ti] (LC) {(1) determine all LC symmetries };
    \node [decision, below of=LC] (decision1) {(2) LC symmetry of order $3$ found?};
    \node [block, left of=decision1, node distance=3cm] (end1) {$\mathcal{U}_G\subset \mathcal{C}_n$};

    \node [decision, below of=decision1, node distance=3.4cm] (decision2) {(3) LC symmetry of order $4$ found?};
    \node [block, left of=decision2, node distance=3cm] (end2) {$\mathcal{S}_G=U_G$};
      \node [decision, below of=decision2, node distance=3.4cm] (decision3) {(4) all LC symmetries from leaves?};
        \node [block, left of=decision3, node distance =3cm] (end3) {$U_G=\left<\mathcal{S}_G\cup L_G\right>$};
            \node [block, right of=decision3, node distance=3cm] (end4) {(5) check for possible non-LC symmetries};
    \path [line] (LC) -- (decision1);
    \path [line] (decision1)  -- node {yes} (end1);
    \path [line] (decision1) -- node {no} (decision2);
    \path [line] (decision2)  -- node {no} (end2);
    \path [line] (decision2) -- node {yes} (decision3);
    \path [line] (decision3) -- node {yes} (end3);
    \path [line] (decision3) -- node {no} (end4);
\end{tikzpicture}
\caption{Algortihm to determine all local unitary symmetries of graph states. The details of how to determine the symmetries in step $(5)$ are provided in the main text. Any symmetry of this type has to be a $2^m$--th root of an LC symmetry of order $4$ (up to multiplication by an element of $\mathcal{S}_G$)\label{algorithm}.}
\end{figure}

Conversely, if not all LC symmetries of order $4$ of the state $\ket{G}$ stem from leaf symmetries, the state can have additional non--LC symmetries (step $(5)$ in Figure \ref{algorithm}). These are of the form $\sigma_1^j\exp(i\alpha_j\sigma_2^j)$ with $\alpha_j=m_j\pi/2^{n_j}$ for $m_j,n_j\in \mathbb{N}$ and wlog $|m_j|\le 2^{n_j}$ (Corollary \ref{cor4}). To determine those symmetries we determine for the equivalence class $W\in U_G/L_G$ of each of those symmetries the representative $V\in W$ constructed in the proof of Corollary \ref{cor2}. This representative $V$ with $V_j\propto \sigma_1^j\exp(i\alpha_j\sigma_2^j)$ for all $j\in\{1,\ldots,n\}$ has several important properties. For all qubits $j$ we have that $\alpha_j=m_j\pi/2^{n_j}$ for $m_j,n_j\in \mathbb{N}$ and wlog $|m_j|\le 2^{n_j}$. By Corollary \ref{bounds} for any qubit $j$ not equivalent to a leaf under local complementation the number of possible values for $n_j$ and thus for $m_j$ is finite. Furthermore, it follows from the considerations in the proof of Corollary \ref{cor2} that even if qubit $j$ is a leaf under local complementation, for this specific representative $V$ the variable $n_j$ also satisfies the bounds specified in Corollary \ref{bounds}.

Thus, in order to find the representative $V$ of any equivalence class $W$, we only have to check a finite number of possible configurations. A systematic way to go through all possibilities is the following. Let us define the set $K\subseteq \{1,\ldots,n\}$ such that for every group of leafs and parent (under local complementation) the set $K$ contains exactly one of these qubits. Furthermore $K$ contains all qubits that do not correspond to a leaf under local complementation. Recall that due to $\ket{G}\in T$ we have that for any qubit $j$ there can only exist one operator $\sigma_2^j$ such that $U_j\propto\sigma_1^j\exp(i\alpha_j\sigma_2^j)$ is a tensor factor of a symmetry (see Sec. \ref{Symmetries of stabilizer states}). Furthermore, as we have already determined the LC symmetries of $\ket{G}$ some of the operators $\sigma_2^j$ are already fixed. Thus, any representative $V$ can be written as
\bea
V=Z^{\vec{k}}\bigotimes_{j\in K}e^{i\alpha_j\sigma_2^j}\label{eq:representative}
\eea
where $\sigma_2^j\in\{X,Y,Z\}$ and for some $j$ these operators are fixed by an LC symmetry of $\ket{G}$ and $\alpha_j=m_j\pi/2^{n_j}$ with $m_j,n_j\in\mathbb{N}$ for all $j\in K$. Furthermore, $n_j$ satisfies the bounds from Corollary \ref{bounds} for all $j$ with $\alpha_j\neq k\pi/2$, $k\in\mathbb{Z}$ and wlog $|m_j|<2^{n_j}$.

Let us mention two observations which can be used to compute the symmetries more efficiently than applying each candidate $V$ of Eq. (\ref{eq:representative}) to $\ket{G}$. First for sufficiently large $l$ the operator $V^l$ (and any operator $(SV)^l$ where $S\in\mathcal{S}_G$) is an LC symmetry of $\ket{G}$ which were already determined at the beginning of the algorithm. All $V$ which do not satisfy this condition can be excluded. Second, there is a systematic way to determine $\sigma_2^j$ (including $\one$) for any qubit $j$. Let $D\subset\{1,\ldots,n\}$ be the subset of qubits for which $\sigma_2^j$ has already been determined. We choose one qubit $l\in D$ and use local complementation as outlined in the proof of Theorem \ref{thm leaf} to transform $\sigma_2^l$ into $Z$. Hence, the resulting symmetry equation reads
\bea
\begin{split}
&Z^{\vec{k}}(e^{i\alpha_l Z_l}\otimes \bigotimes_{m\in D\setminus \{l\}}e^{i\alpha_m \sigma_2^m}\otimes\\ & \bigotimes_{m\not\in D}e^{i\alpha_m \sigma_2^m}\ket{G'})\propto\ket{G'},\end{split}\label{eq:symeq}
\eea
where we used that $P\ket{G}\propto Z^{\vec{k}}\ket{G}$ for some $\vec{k}$ and any $P\in\mathcal{P}_n$ and $\ket{G'}$ denotes the graph state after local complementation. Projecting now qubit $l$ onto the state $\ket{0}$ reduces Eq. (\ref{eq:symeq}) to a similar equation for $n-1$ qubits. This is due to the fact that any graph state can be written as $\ket{G}=\ket{0}_l\ket{\tilde{G}}+\bigotimes_{k\in N_l}Z_k\ket{1}_l\ket{\tilde{G}}$, where $\ket{\tilde{G}}$ is a graph state of $n-1$ qubits. Repeating this step for each of the qubits in $D$ leads to the equation 
\bea
Z^{\vec{k}'}\bigotimes_{m\not\in D}e^{i\alpha_m \sigma_2^m}\ket{G''}\propto\ket{G''}
\eea
for some graph state $\ket{G''}$. Note that the symmetries of $\ket{G''}$ coincide with the previously undetermined part of the potential symmetries of $\ket{G}$ up to LC operators which resulted from the local complementation (and are known). As local complementation preserves the order of LC operators of order $4$ it only remains to determine the LC symmetries of order $4$ of $\ket{G''}$ which leads to some $\sigma_2^m$ for some $m\not\in D$. In case the operator $\sigma_2^m$ is determined for all $m$ the whole symmetry can be easily computed. Otherwise the last step in the algorithm is repeated.

\section{Examples\label{examples}}
In this section we present some examples for graph states with additional symmetries i.e. with $\mathcal{S}_G\subsetneq U_G$. These examples illustrate the variety of possible symmetries and are meant to give an overview of the structures local unitary symmetries can take. We focus here on discrete symmetries as leaves and the structures related to leaves under local complementation have already be discussed in Sec. \ref{not and pre}. Table \ref{tablegraphs} in Appendix \ref{tablegraphssection} gives a set of generators for the additional discrete symmetries of the respective graph state.

Graph a) has an LC symmetry of order $3$ which is a product of its LC symmetries of order $4$. In contrast to that, graph $b)$ only allows for LC symmetries of order $3$. Thus, the existence of LC symmtries of order $3$ does not imply the existence of LC symmetries of order $4$. Graph $c)$ and $d)$ are logical states of two instances of the quantum Reed Muller codes \cite{AMS99}. These codes admit a nontrivial diagonal transversal gate (see section \ref{Local symmetries and transversal gates}). The graph states corresponding to the logical $\ket{0}_L$ (and $\ket{1}_L$) of these codes have an additional discrete symmetry, a transversal $T$ gate. Note that graph $d)$ admits, besides several LC symmetries of order $4$, a local symmetry $U$ (root of one of the LC symmetries of order $4$) that is no LC symmetry and the graph does not contain a leaf under local complementation. For this symmetry we have that $\alpha_j=\pi/8$ for all $j\in\{1,\ldots,n\}$ in agreement with Theorem \ref{thm leaf}.

Furthermore, observe that for graphs $c)$ and $d)$ we have that $\sigma_2^j\in\{X,Z\}$ for all $j$. However, it is also possible to have symmetries where $\sigma_2^j=X$ holds for all $j$. Examples for this case are graphs $e)$ and $f)$. As pointed out by \cite{RoGo19}, any graph consisting of two copies of a complete graph of an even number of vertices with edges between corresponding vertices gives an example for such a state. Note that in graph $f)$ vertices $2$ and $3$ are twins and thus the graph also has a leaf symmetry. It is also possible to have $\sigma_2^j=Y$ for all $j$ and graph $a)$ is an example for that. Finally, let us mention, that $Z^{\otimes n}$ can never be an element of $\mathcal{S}_G$ for any graph $G$. Hence, there exists no symmetry with $\sigma_2^j=Z$ for all $j$ (Theorem \ref{conditionlc4}).
\begin{figure}
\begin{minipage}{0.23\textwidth}
\includegraphics[width=\textwidth]{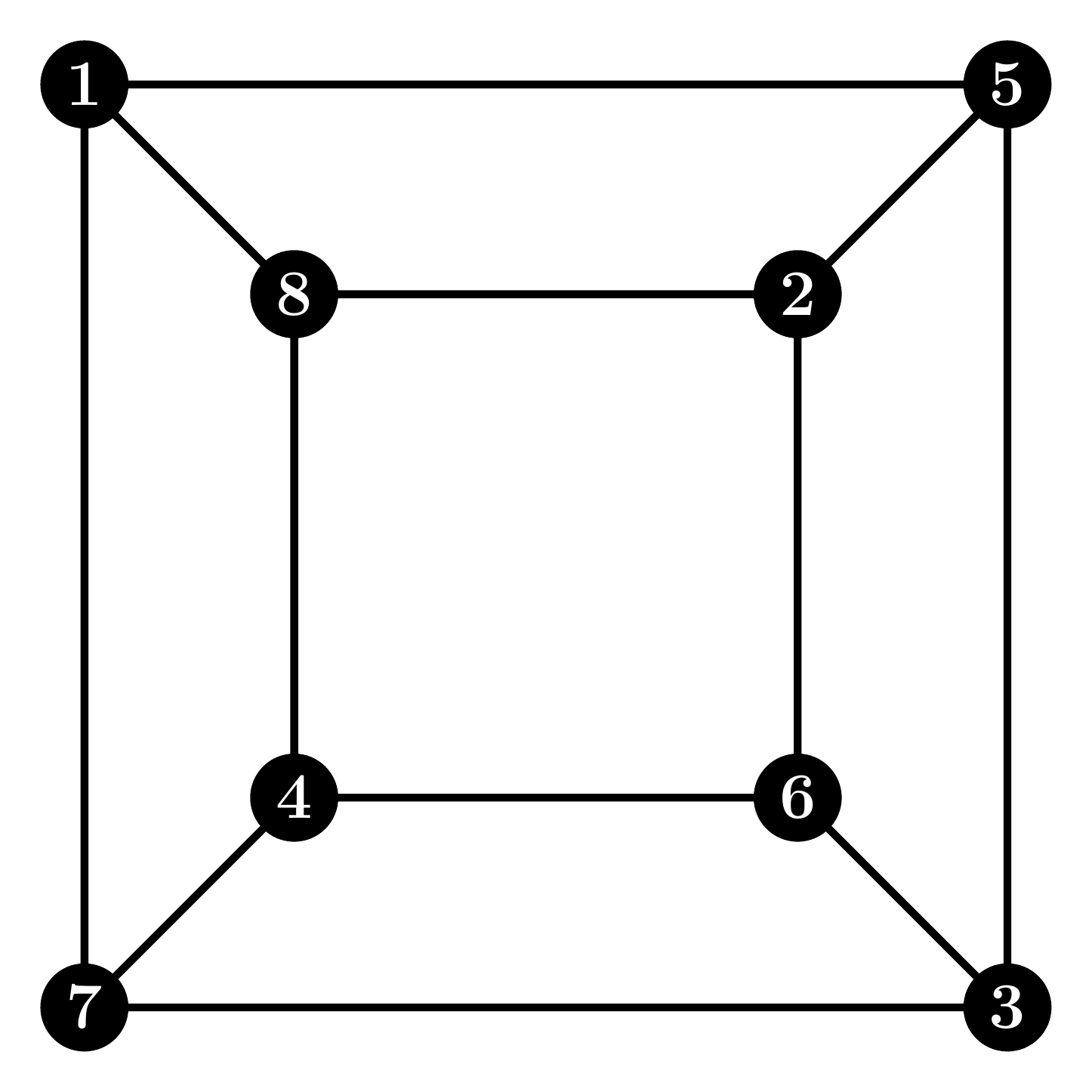}
a)
\end{minipage}
\hfill
\begin{minipage}{0.23\textwidth}
\includegraphics[width=\textwidth]{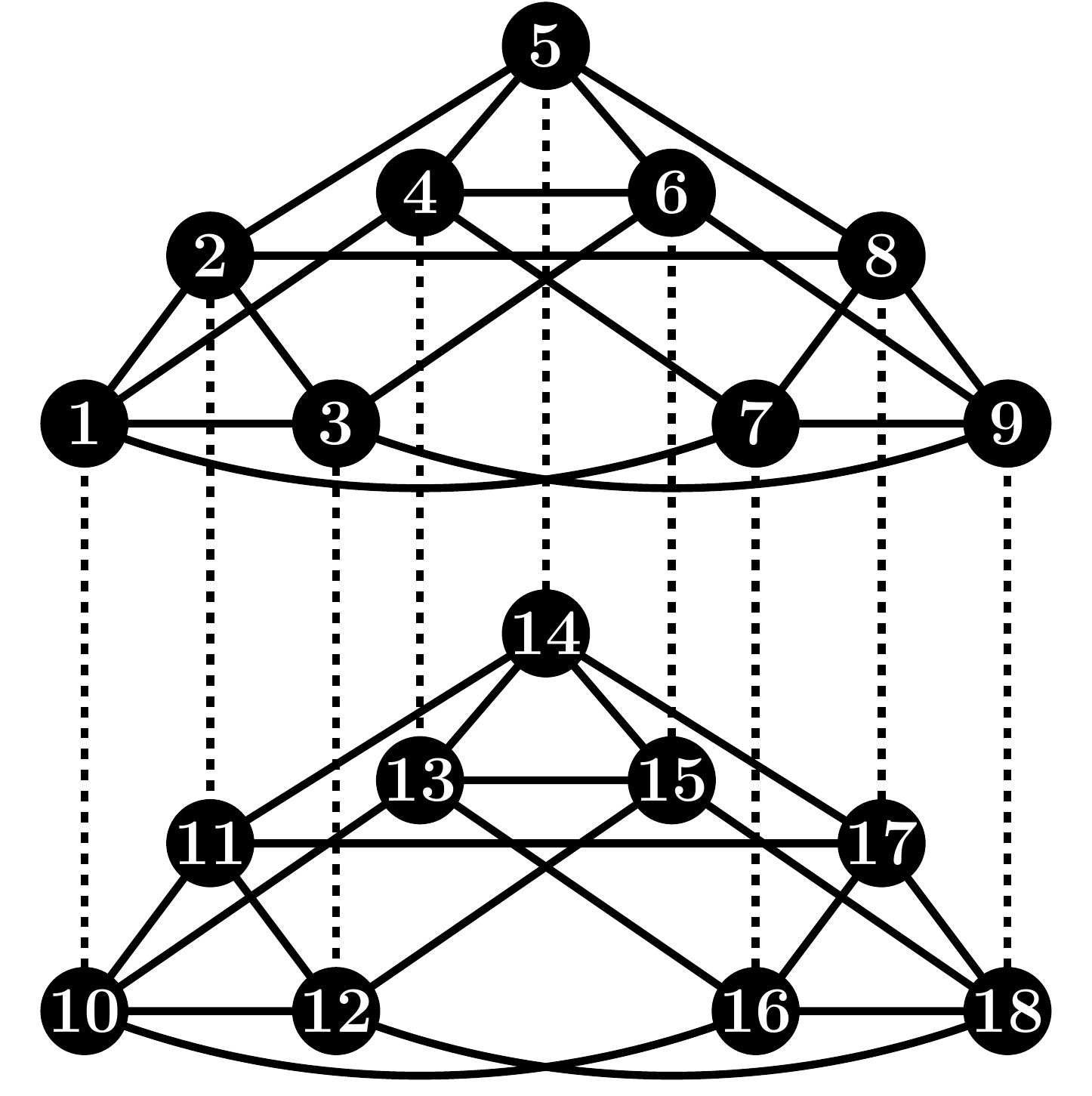}
b)
\end{minipage}
\begin{minipage}{0.23\textwidth}
\includegraphics[width=\textwidth]{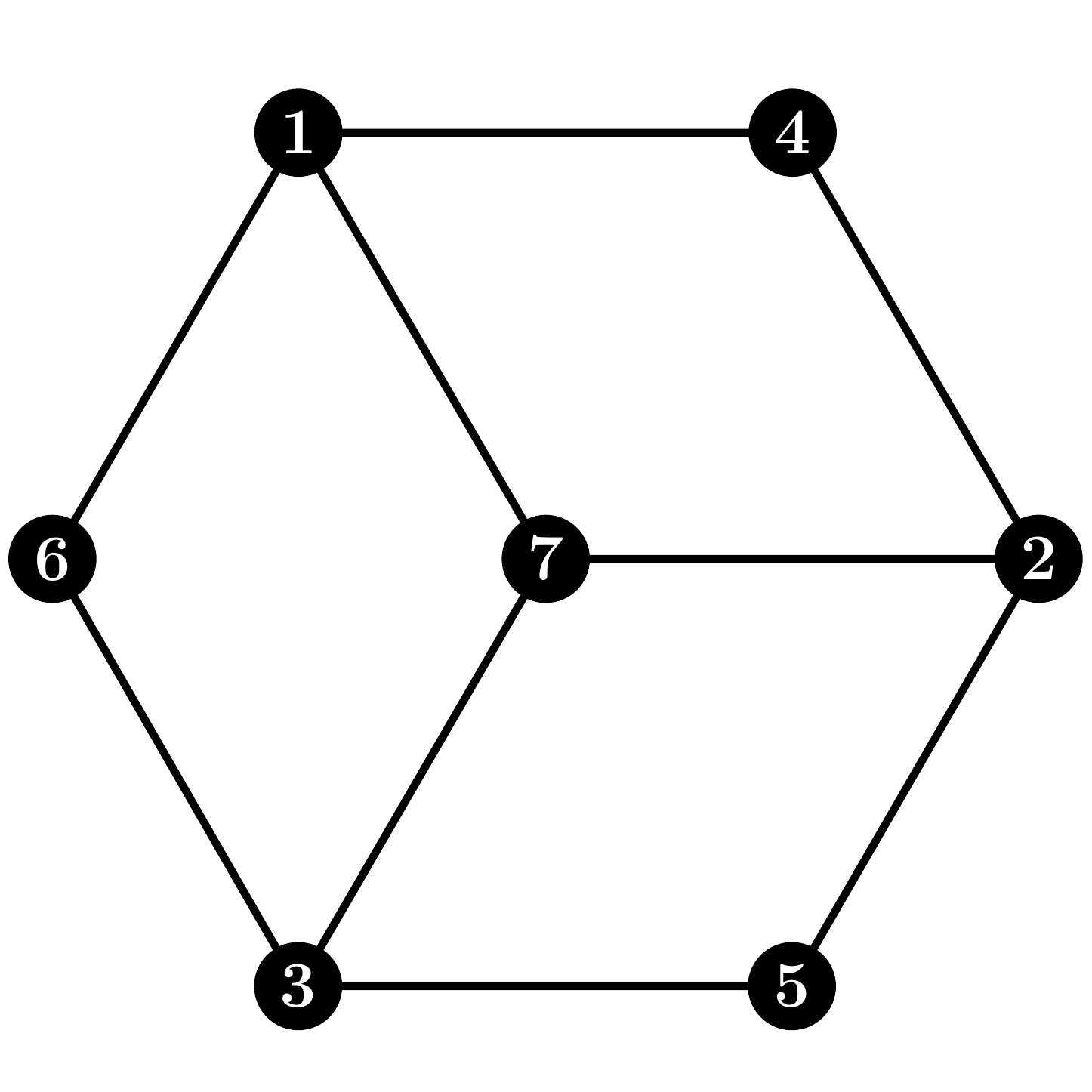}
c)
\end{minipage}
\hfill
\begin{minipage}{0.23\textwidth}
\includegraphics[width=\textwidth]{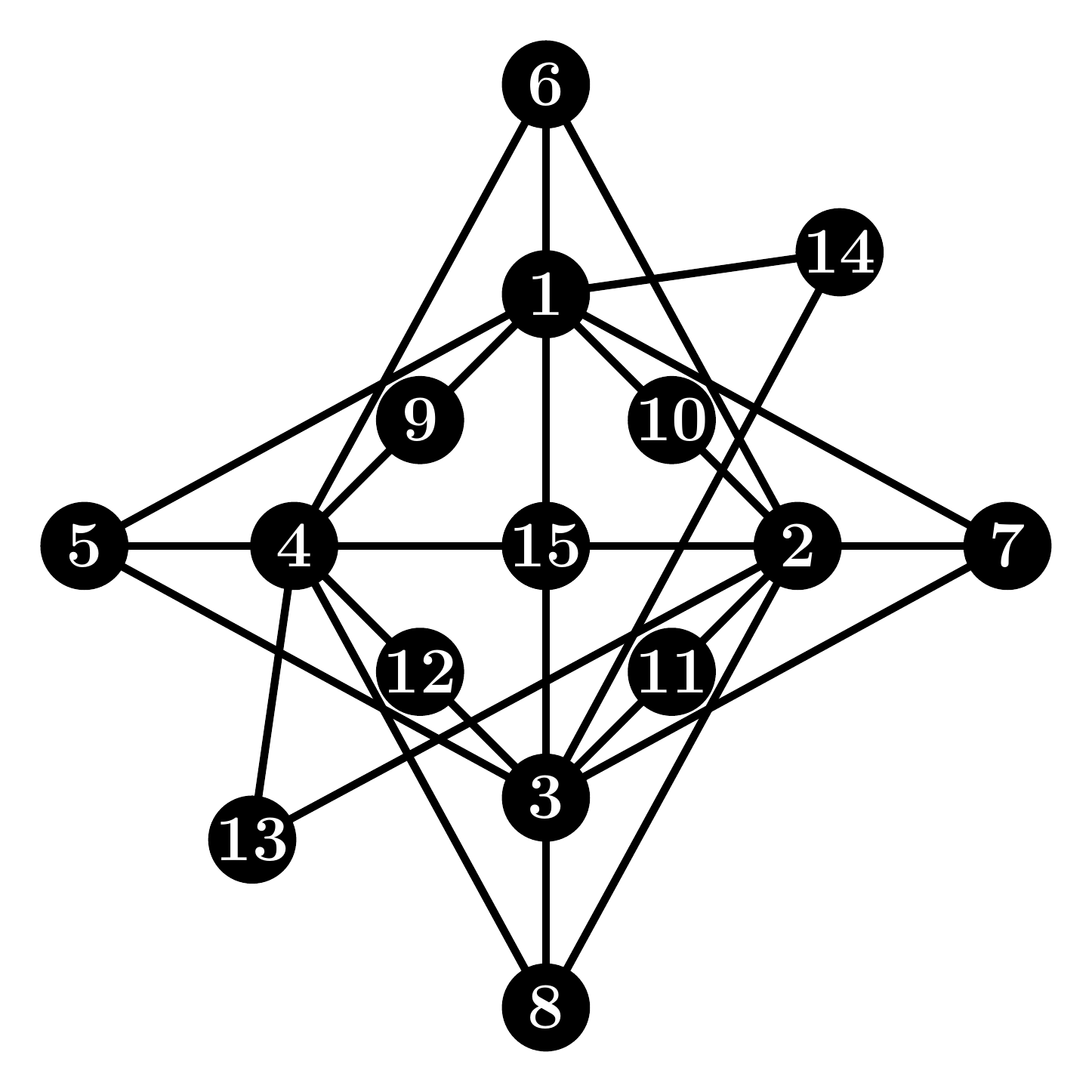}
d)
\end{minipage}
\begin{minipage}{0.23\textwidth}
\includegraphics[width=\textwidth]{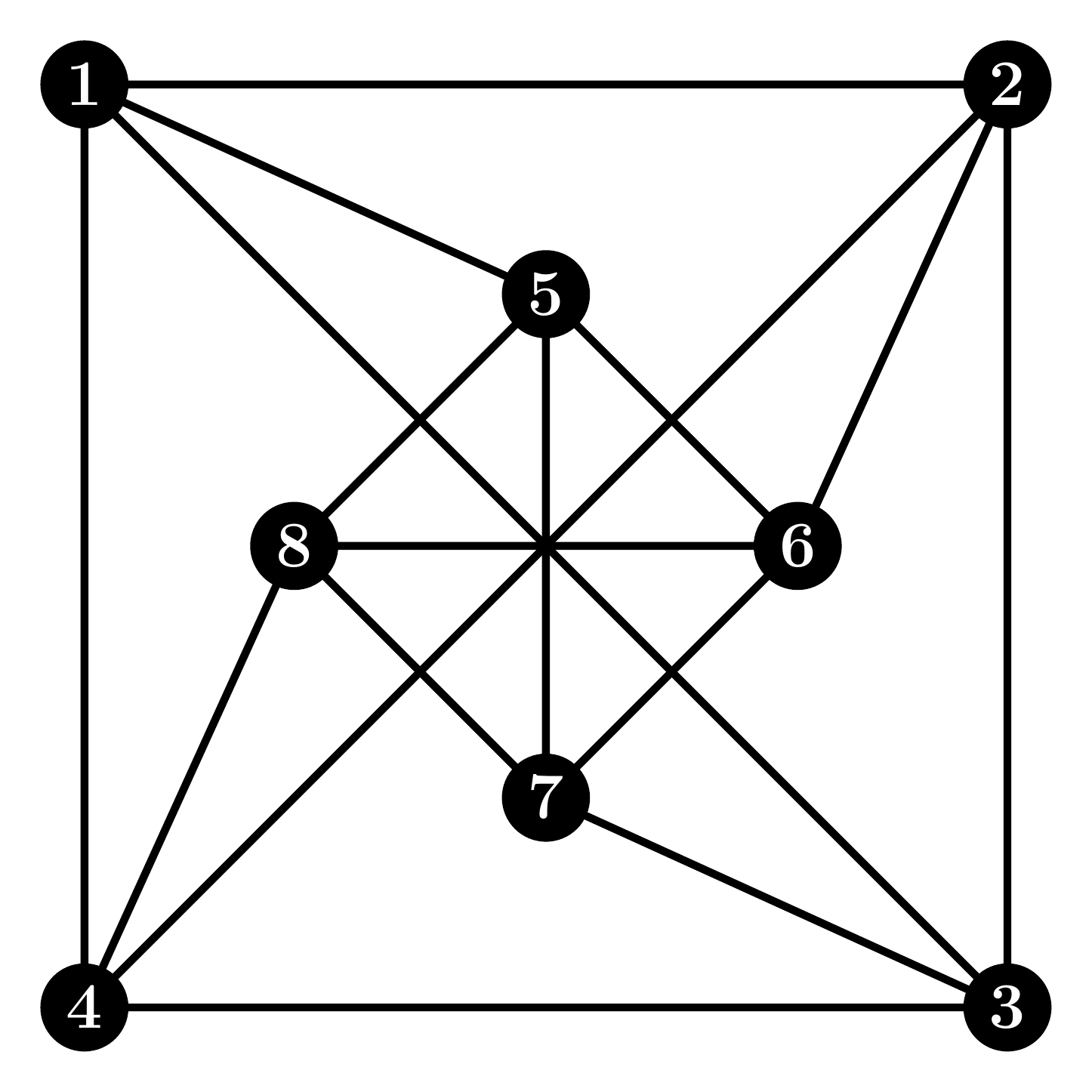}
e)
\end{minipage}
\begin{minipage}{0.23\textwidth}
\includegraphics[width=\textwidth]{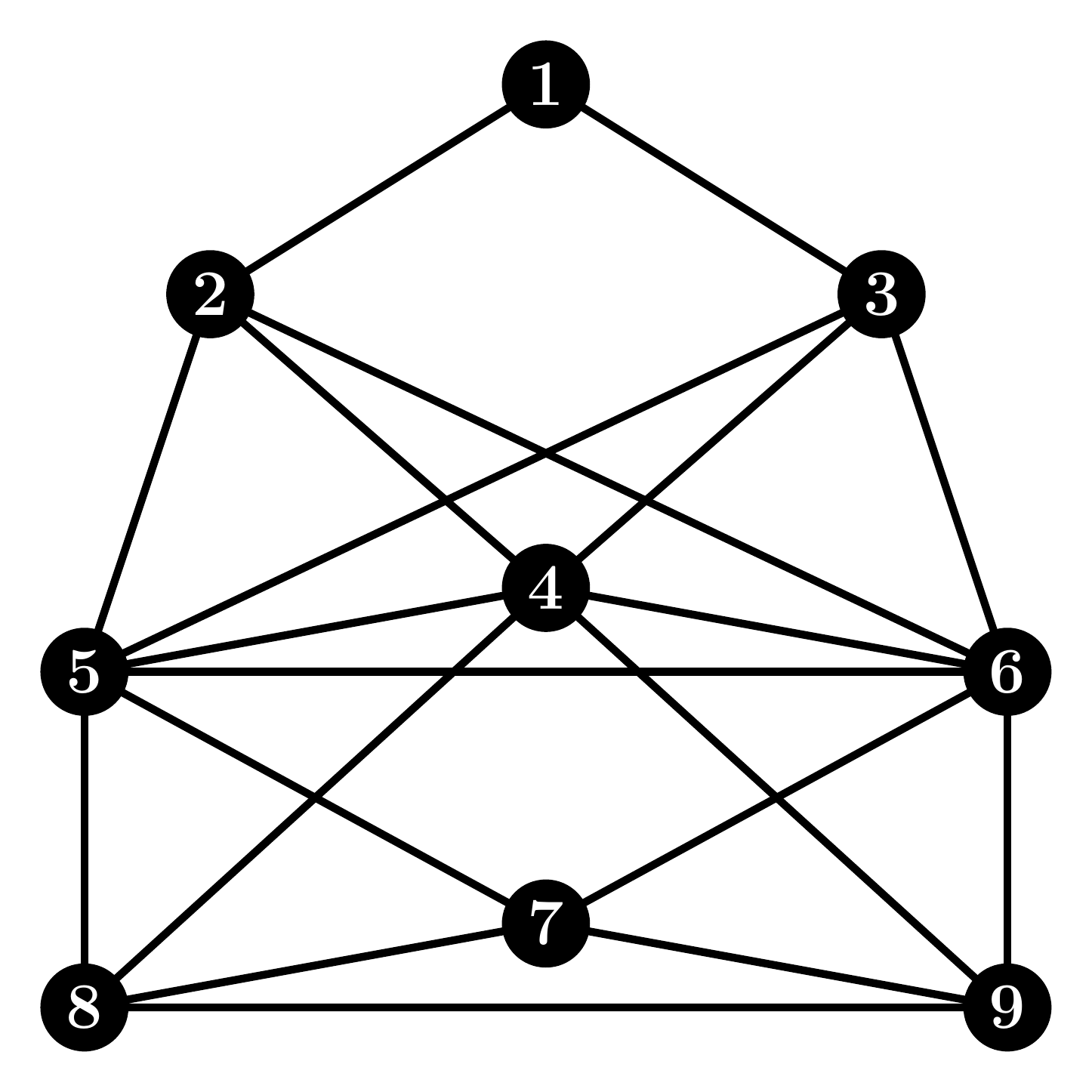}
f)
\end{minipage}
\caption{Examples for graphs corresponding to graph states with additional discrete symmetries (see main text and Table \ref{tablegraphs}). Note that graph $f)$ also has an additional continuous symmetry on the twin vertices $2$ and $3$.}
\label{graphs}
\end{figure}

\section{Applications\label{applications}}

In this section we discuss several applications of the additional symmetries of stabilizers states. We first show that they find applications in fault tolerant quantum computing \cite{GoDa97,CE17} and then study their relevance in entanglement theory.

\subsection{Local symmetries and transversal gates\label{Local symmetries and transversal gates}}
The goal of quantum error correction is to protect a logical qubit from errors by introducing redundancy and storing it using several physical qubits. To use the logical qubits for computations one needs to be able to perform gates on the logical level. For any nontrivial quantum error correcting code any logical gate will act on at least two physical qubits nontrivially. If such a logical gate acts non-locally it can cause the errors to spread on the physical qubits and thus eventually cause a logical error. Thus, it is desirable to have codes that admit at least some number of logical gates that are local operations on the physical level \cite{CE17}. In the context of quantum error correction and fault-tolerant quantum computing such gates are called transversal gates. Note that it has been shown that there exists no nontrivial code that admits a complete set of transversal gates \cite{CX08,BE09}.

There is a close connection between transversal gates for stabilizer codes and local symmetries of stabilizer states. In particular, if a stabilizer code encoding one logical qubit has a diagonal traversal gate $T$, i.e. $T\ket{0}_L\propto \ket{0}_L$ and $T\ket{1}_L\propto \ket{1}_L$, then $T$ is a local symmetry of $\ket{0}_L$. In turn, if we start from a graph state $\ket{G}$ with an additional symmetry $U\in U_G$, $U\not\in \mathcal{S}_G$, $U\not \in \mathcal{C}_n^3$ then we can construct a stabilizer code with a transversal diagonal gate $U$ as we show subsequently \footnote{In case $\alpha_0=m_0\pi$ the operator $U$ is the logical identity for the stabilizer code constructed in the following.}. Let us consider a symmetry of $\ket{G}$ such that 
\bea
U\ket{G}=\bigotimes_j \sigma_1^j e^{i\alpha_j \sigma_2^j}\ket{G}=e^{i\alpha_0}\ket{G}
\eea
with $\alpha_0\neq m_0\pi$, $m_0\in \mathbb{Z}$. The stabilizer code then consists of the logical states $\ket{0}_L=\ket{G}$ and $\ket{1}_L=P\ket{G}$, where $P\in\mathcal{P}_n$ is chosen as follows. For those $k\in\{1,\ldots n\}$ for which $\sigma_1^k=\one$ we choose $P_k$ such that $[P_k,\sigma_2^k]\neq 0$ \footnote{Note that we always find such a party if we consider instead of $U$ a symmetry $SU$ for a suitable $S\in\mathcal{S}_G$.}. Note that this condition is fulfilled by two Pauli operators for any party $k$ and thus we can ensure that $\left <0|1\right>_L=0$. This is due to the fact that $\bra{G}P\ket{G}=0$ for any $P$ not proportional to an element of $\mathcal{S}_G$. For those $k\in\{1,\ldots n\}$ for which $\sigma_1^k\neq\one$ we choose $P_k =\one$. As $U\ket{1}_L=\pm PU^\dagger\ket{0}_L=\pm e^{-i\alpha_0}P\ket{0}_L=\pm e^{-i\alpha_0}\ket{1}_L$ we have

\begin{align}
U\ket{0}_L&=e^{i\alpha_0} \ket{0}_L\\
U\ket{1}_L&=\pm e^{-i\alpha_0} \ket{1}_L.
\end{align}

Thus, $U$ is a transversal gate for the constructed stabilizer code. Using this relation between local symmetries of stabilizer states and transversal gates of stabilizer codes, we see that Theorem \ref{thm leaf} also follows from the results on transversal gates of stabilizer codes in \cite{AJ16}.

One can easily construct new stabilizer codes with transversal gates from known ones as follows. Let $\ket{G}$ be the graph state presented in Fig. \ref{graphs}$d)$. Attaching the qubit $n+1$ to any qubit with $\sigma_2^j=Z$ leads to a new graph state with the same symmetry on the original subgraph (the first $n$ qubits). Defining the stabilizer code as explained above leads to a code with a transversal gate of the form $U\otimes \one$. Note that this construction does not work for a qubit $j$ with $\sigma_2^j=X,Y$.

\subsection{Separable transformations \label{sep trafo sec}}
As mentioned in the introduction, entanglement is a resource under local operations and classical communication (LOCC). If a state $\ket{\psi}$ can be deterministically transformed into a state $\ket{\phi}$ via LOCC then $\ket{\psi}$ is at least as entangled as $\ket{\phi}$ with respect to any entanglement measure. Thus, LOCC introduces a partial order on the Hilbert space and characterizing possible LOCC transformations is crucial for identifying states with useful entanglement properties. However, due to the intricate structure of LOCC with possibly infinitely many rounds of classical communication, sometimes the larger set of separable transformations is considered, which has a simpler mathematical description. In \cite{HeMa19} we showed that in order to decide whether a separable transformation among two fully entangled pure states is possible it is not sufficient to consider local invertible Kraus operators. We call the latter set of operations SEP1 in the following \cite{GoGi11,HeMa19}. In \cite{HeMa19} we used the results presented here to construct the first example of a state transformation which is possible via SEP, but not via SEP1. In the following we present a general construction how to find examples of transformations among fully entangled pure states which are possible via SEP but not via SEP1.

Consider a graph state $\ket{G}$ with no additional symmetries, i.e. $U_G=\mathcal{S}_G$. Note that any state which does not solve Eq. (\ref{eq:lc31}) or Eq. (\ref{eq:lc41}) to (\ref{eq:lc43}) has this property. In \cite{GoGi11} it was shown that a state $g \ket{\psi}$ can be transformed into a state $h \ket{\psi}$ via SEP1 with $g=g_1\otimes \ldots\otimes g_n$ and $h=h_1\otimes \ldots\otimes h_n$ invertible if and only if there exist symmetries $U_{(k)}\in U_\psi$ and probabilities $p_k\ge 0$, $\sum_k p_k=1$ such that\bea
\sum_kp_k(U_{(k)})^\dagger H U_{(k)}=\frac{||g\ket{\psi}||}{||h\ket{\psi}||}G,\label{eq:sep}
\eea where $H=h^\dagger h$ and $G=g^\dagger g$. We consider a transformation from $\ket{G}$ to a state $h\ket{G}$, i.e. $G=\one$. Using that $\mathcal{S}_G$ is abelian, we find that Eq. (\ref{eq:sep}) is fulfilled only if $\tr(H S)=0$ for any $S\in\mathcal{S}_G\setminus\{\one\}$ \cite{RaBr17}. Thus, choosing $H=\one\otimes\bigotimes_{k\in \text{supp}(S_{(j)})} (\one+a (S_{(j)})_k)$ for some $a\in (0,1)$ and some canonical generator $S_{(j)}$ the transformation is not possible as $\tr(HS_{(j)})\neq 0$.

Let us now construct a SEP transformation to transform $\ket{G}$ into $h\ket{G}$. Recall that in Lemma \ref{annihilate} we construct for a given stabilizer state local operators $Q_S^{f}$, based on elements $S$ of its stabilizer, that annihilate the state. Note that $f$ has to satisfy $\sum_{m\in \text{supp}(S)}f(m)\text{ mod } 2=1$ and thus there are $2^{|\text{supp}(S)|-1}$ different functions $f$ which we label by $f_k$. Using projectors $Q_{S_{(j)}}^{f_k}$ and $n_j=|\text{supp}(S_{(j)})|$, the Kraus operators for the separable map are \bea
M_k&=&\sqrt{\frac{a^{n_j}/2^{2n_j-1}}{(1+a^{n_j})(1+a)^{q_k}(1-a)^{n_j-q_k}}} h Q_{S_{(j)}}^{f_k}
\eea
for $k\in\{1,\ldots,2^{n_j-1}\}$, where $q_k=|\{j|f_k(j)=0\}|$, and 
\bea
M_k&=&\frac{1}{\sqrt{2^{n_j-1}(1+a^{n_j})}} h P_{(k)}
\eea 
for $k\in\{2^{n_j-1}+1,\ldots,2^{n_j} \}$, where $P_{(k)}$ denotes all elements from the group $\left<\{S_{(l)}|l\in N_j\}\right>$ for all $k$ and $S_{(j)}$ is the canonical generator corresponding to qubit $j$. Note that, since the stabilizer is abelian, the subgroup $\left<\{S_{(l)}|l\in N_j\}\right>$ has exactly $2^{n_j-1}$ different elements. It is straightforward to verify the completeness relations $\sum_k M_k^\dagger M_k=\one$ and that the separable map corresponding to these Kraus operators implements indeed the transformation.

\subsection{LOCC$_\N$ transformations}
In this section we show that additional local symmetries for graph states, i.e. those not contained in the stabilizer, allow for finite round LOCC transformations (LOCC$_\N$) which are not possible if only stabilizer symmetries are utilized. 

Let us first recall the necessary and sufficient condition for reachability of a state via LOCC$_\N$ \cite{SpCo17}. Let $\ket{\psi}$ be a state with a finite, unitary symmetry group $U_\psi$. Then, a state $h\ket{\psi}$ in its SLOCC class is reachable via LOCC$_\mathbb{N}$ iff there exists a $U\in U_\psi$ such that (up to permutations of the qubits)
\begin{align}
[H_1,U_1]&\neq 0\label{loccn1}\\
[H_j,U_j]&=0\ \forall j\in\{2,\ldots,n\}\label{loccn2}.
\end{align}
The states from which $h\ket{\psi}$ is reachable are given by $g\ket{\psi}$, where $G=g^\dagger g$ is such that $G_1=pH_1+(1-p)U_1^{\dagger} H_1 U_1$ for some $p\in(0,1)$ and $G_j=H_j$ for all $j\ge 2$. Hence, we see that if a stabilizer state has a local symmetry which is diagonal in a different basis than any of the elements of its stabilizer new states are reachable. Moreover, if if this local symmetry is diagonal in the same basis as an element of the stabilizer the reachable states stay the same but more transformations become possible. Suppose a graph state has an LC symmetry of order $3$ which, up to conjugation by local Cliffords, is of the form $\exp(i\pi/4X)\exp(i\pi/4 Y)$ on every qubit. Then, for instance, any state $h\ket{G}$ with $H_1$ not diagonal in the eigenbasis of $\exp(i\pi/4X)\exp(i\pi/4 Y)$ and $H_j$ diagonal in this basis for $j\ge 2$ is reachable as Eq. (\ref{loccn1}) and (\ref{loccn2}) are satisfied by this additional symmetry. However, there does not exist any Pauli operator in $\mathcal{S}_G$ which is diagonal in the eigenbasis of $\exp(i\pi/4X)\exp(i\pi/4 Y)$ for all but one qubit \footnote{This holds as we only consider fully connected graphs and thus there exists no stabilizer element which acts nontrivial on only one qubit. Furthermore, conjugation by local Clifford operators does not change the support of an operator.}. A similar construction also works for the case of LC symmetries of order $4$. However, the tensor product of the operators in the exponent of each factor is an element of the stabilizer for these symmetries (Theorem \ref{conditionlc4}). Thus, if every local operator of the symmetry is of the form $U_j\propto\exp(\pm i\pi/4 \sigma_2^j)$ this symmetry has the same commutation properties as the corresponding stabilizer and by itself does not allow to reach more states than the stabilizer. Nevertheless, since $U$ is not an element of the stabilizer, new transformations are possible, as there exist more states $g\ket{\psi}$ which can reach $h\ket{\psi}$ by LOCC$_\mathbb{N}$. Furthermore, in this case new reachable states can be found when considering products of the LC symmetry of order $4$ and elements of the stabilizer. We conclude that graph states (and thus stabilizer states) with additional (discrete) symmetries can be more powerful regarding LOCC$_\mathbb{N}$ transformations than graph states with $\mathcal{S}_G=U_G$ \footnote{This is decided by the specific form of the additional symmetries. The statement holds for any example presented in Sec. \ref{examples}.}.

The volume of the set of states reachable by an initial state via LOCC is an entanglement measure, called accessible entanglement \cite{ScKa15}. To compute this entanglement measure for a state one has to determine all states reachable from this state by LOCC (not just LOOC$_\mathbb{N}$ as considered above). A common approach to this problem is to use SEP convertibility as necessary condition in order to gain insights on the LOCC convertibility of a state. However, as Eq. (\ref{eq:sep}) is not necessary for SEP convertibility and graph states were used to show this (see discussion above and \cite{HeMa19}), new methods to determine all possible LOCC transformations might need to be developed.

\section{Conclusion and Outlook\label{Conclusion and Outlook}}
In this work we have investigated the local (invertible and non-invertible) symmetries of fully entangled stabilizer states. We have characterized all local invertible symmetries of stabilizer states and have provided an algorithm which determines them. To this end we have used that stabilizer states are critical states and thus local, non-unitary symmetries are determined by the local unitary ones. Furthermore, every stabilizer state is LC equivalent to a graph state, and thus, in order to determine all LC symmetries of stabilizer states, it is sufficient to consider graph states. 

We have shown that there are two different types of graph states, those which possess an LC symmetry of order $3$ and those which do not have such a symmetry. The symmetry group $U_G$ of the first type of graph state has to be a subgroup of the local Clifford group $\mathcal{C}_n$ and is therefore discrete and finite. Graph states of the second type have a continuous unitary symmetry if and only if the corresponding graph has a leaf (up to local complementation). Note that this is the only case in which a graph state also has a symmetry in $GL$. Any other symmetry a graph state of the second type can have is of the form $U_j\propto \sigma_1^j\exp(i\alpha_j\sigma_2^j)$ where $\alpha=m_j\pi /2^{n_j}$, $m_j,n_j\in \mathbb{N}$, $\sigma_1^j\in \mathcal{P}_1$ and $\sigma_2^j\in\{X,Y,Z\}$ for all $j$. Moreover, the number $n_j$ is bounded by the number of neighbours of qubit $j$ (or a neighbour of qubit $j$) and thus a graph state can only have finitely many additional discrete symmetries. We have shown that any of these additional discrete symmetries is the $2^k$--th root of a LC symmetry of order $4$ (up to multiplication with an element of $\mathcal{S}_G$). Combing this with the results on graph states with LC symmetries of order $3$ we have concluded that any graph state with an additional local symmetry has an LC symmetry. Furthermore, we have provided necessary and sufficient conditions on the adjacency matrix of a graph for the existence of LC symmetries.

We have discussed applications of the results in fault-tolerant quantum computing and entanglement theory. In particular, we have provided a general construction for SEP transformations among pure states that are not possible with solely Kraus operators of full rank. Furthermore, the relevance of these additional symmetries for transformations using LOCC$_\mathbb{N}$ has been demonstrated. In the future it will be interesting to identify new, more practical applications of stabilizer states, which are based on the additional symmetries \cite{RaRo19}. Similar investigations of more general states such as LME states \cite{LME} and, in particular, hypergraph states \cite{Hyper} and higher dimensional stabilizer states might shine new light on their entanglement properties and applications thereof.

We acknowledge preliminary research of D. Sauerwein and R. Brieger \cite{RaBr17}, especially their results on the transformations via SEP1 used in Sec. \ref{sep trafo sec}. We acknowledge financial support from the Austrian Science Fund (FWF) grant DK-ALM: W1259-N27 and the SFB BeyondC  (Grant  No.   F7107-N38) and of the Austrian Academy of Sciences via the Innovation Fund “Research, Science and Society”.

\appendix

\section{Semi Clifford operators \label{appendixa}}
As mentioned in the main text, a unitary operator $U$ is called semi Clifford operator if it maps at least one Pauli operator to a Pauli operator (up to a phase). Let us show that any such operator is of the form
\bea 
U\propto C e^{i \alpha \sigma},
\eea
where $C\in \mathcal{C}_1$, $\alpha \in \R$ and $\sigma\in \{X, Y, Z\}$. 

\begin{proof}
There exists a pair of Pauli operators $\sigma_k,\sigma_l$ such that
\bea
U\sigma_kU^\dagger\propto\sigma_l.\label{sym equ}
\eea
Since $U$ is unitary, we can write it in its Euler decomposition. If $\sigma_k=\sigma_l$ we parametrize the operator as $U\propto e^{i\alpha\sigma_k}e^{i\beta\sigma_j}e^{i\gamma\sigma_k}$ where $j\neq k$ and $\alpha,\beta,\gamma\in \R$. Using this decomposition in Eq. (\ref{sym equ}) together with the commutation relations of the Pauli operators we find that $e^{2i\beta\sigma_j}\propto\one$ and thus $2\beta=n\pi$, $n\in\mathbb{Z}$. Then, $U\propto \pm i\sigma_je^{i(\gamma-\alpha)\sigma_k}$ and, thus, $U$ is of the form claimed above. If  instead $\sigma_k\neq \sigma_l$ we parametrize the operator as $U\propto e^{i\alpha\sigma_k}e^{i\beta\sigma_l}e^{i\gamma\sigma_k}$ where $\alpha,\beta,\gamma\in \R$. Using this decomposition again in Eq. (\ref{sym equ}) we find $e^{2i\beta \sigma_l}\sigma_k\propto e^{-2i\alpha \sigma_k}\sigma_l$. Multiplying this equation with $\sigma_k$ ($\sigma_l$) and computing the trace we find that $2\alpha=\frac{\pi}{2}+n\pi$ and $2\beta=\frac{\pi}{2}+m\pi$ with $m,n\in\mathbb{Z}$. Hence, $U$ is again of the form claimed above. 
\end{proof}

\section{Proof to Theorem \ref{generalform}\label{appendixb}}

Here we prove Theorem \ref{generalform} of the main text. It restricts the general form a local unitary symmetry of a stabilizer state. In order to improve readability let us restate the theorem.\newline

\noindent {\bf Theorem 2.}  \textit{Let $\ket{\psi}\in\left(\mathbb{C}^2\right)^{\otimes n}$ be a fully entangled stabilizer state and let $U\in U_\psi$ be a local symmetry of $\ket{\psi}$. Then \begin{equation}
U_j\propto \left\{\begin{matrix}C_j \\\sigma_1^j\exp\left(i\alpha_j\sigma_2^j\right)
\end{matrix}\right.
\end{equation}
\indent with $C_j\in\mathcal{C}_1^3$, $\alpha_j\in\mathbb{R}$, $\sigma_1^j\in\{\one,X,Y,Z\}$ and $\sigma_2^j\in\{X,Y,Z\}$ for all $j\in\{1,\ldots,n\}$.
}

\begin{proof}
Let $\ket{\psi}\in\left(\mathbb{C}^2\right)^{\otimes n}$ be a stabilizer state and $U\in U_\psi$ be a local symmetry of $\ket{\psi}$. By Theorem \ref{theorem semi clifford} each local factor $U_j$, $j\in\{1,\ldots,n\}$ is a semi Clifford operator and thus by Eq. (\ref{eqn: form semi clifford}) it is of the form $U_j\propto C e^{i\alpha\sigma_a}$ for some $C\in\mathcal{C}_1$ \footnote{The local operators can be different. However, as the proof holds for any local factor we will not write the index whenever it does not lead to any confusion.}. As $U_\psi$ forms a group also $U^2\in U_\psi$ has to be a symmetry of $\ket{\psi}$ and thus also $U_j^2$ has to be a semi Clifford operator. We have\bea
U_j^2\propto  C e^{i\alpha\sigma_a} C e^{i\alpha\sigma_a}=C^2e^{\pm i\alpha\sigma_b} e^{i\alpha\sigma_a},
\eea 

where we used that $C$ is a Clifford operator and $C^\dagger \sigma_a C=\pm\sigma_b$. Thus, $U_j^2$ is a semi Clifford operator iff $e^{\pm i\alpha\sigma_b} e^{i\alpha\sigma_a}$ is a semi Clifford, i.e. iff there are Pauli operators $\sigma_c,\sigma_d\in\{X,Y,Z\}$ such that\bea
e^{\pm i\alpha\sigma_b} e^{i\alpha\sigma_a}\sigma_c(e^{\pm i\alpha\sigma_b} e^{i\alpha\sigma_a})^\dagger=\pm \sigma_d
\eea
Let us now analyze all possibilities for the Pauli operators occuring in this equation. If $\sigma_a=\sigma_c=\sigma_b$ then $C^\dagger \sigma_a C=\pm\sigma_a$ and thus $C=\sigma_ee^{i\beta\sigma_a}$, where $\sigma_e\in\{X,Y,Z\}$ and $\beta\in\{0,\pm\pi/4\}$. Then, $U_j$ is of the form claimed in the theorem. If $\sigma_a=\sigma_c\neq \sigma_b$ then we find $e^{\pm 2i\alpha\sigma_b}\sigma_c=\pm\sigma_d$ and thus $\alpha=k\pi/4$, $k\in\mathbbm{Z}$. This implies that $U_j\in\mathcal{C}_1$. Finally, let us consider the case $\sigma_a\neq\sigma_c$. If additionally $\sigma_d=\sigma_b$ we have $e^{ 2i\alpha\sigma_a}\sigma_c=\pm\sigma_d$ and thus $\alpha=k\pi/4$, $k\in\mathbbm{Z}$. In case $\sigma_d\neq\sigma_b$ we have $e^{ 2i\alpha\sigma_a}\sigma_c=\pm e^{\mp 2i\alpha\sigma_b} \sigma_d$. Multiplying both sides of the latter equation with $\sigma_c$ and computing the trace we find that if $2\alpha\neq k\pi/4$ then $\sigma_c=\sigma_d$ and consequently $\sigma_a=\sigma_b$ has to hold. If $2\alpha= k\pi/2$ we have that $U_j\in \mathcal{C}_1$. If $2\alpha=\pi/4+k\pi/2$ we again multiply both sides of the equation with $\sigma_c$ and compute the trace. The resulting equation allows for two different solutions. The first one is that $\sigma_c=\sigma_d$ and thus $\sigma_a=\sigma_b$. The other possiblity is that $\sigma_c\propto \sigma_b\sigma_d$ and consequently $\sigma_a\sigma_c\propto\sigma_d$. However, this implies that $\sigma_c\propto \sigma_a\sigma_b\sigma_c$ and thus $\sigma_a=\sigma_b$. Thus, in all of these cases we find $U_j\in \mathcal{C}_1$ (by the same argument as above), which completes the proof.
\end{proof}

\section{Proof Lemma \ref{LC3oneLC3all}\label{AB}}
We provide here a proof for Lemma \ref{LC3oneLC3all}. In order to improve readability let us restate the lemma.\newline

\noindent {\bf Lemma 1.}  \textit{Let $\ket{G}$ be a graph state on $n$ qubits and let $U\in U_G$ be such that $U_1\in\mathcal{C}_1^3$. Then $U_j\in \mathcal{C}_1^3$ for all $j\in\{1,\ldots,n\}$ and any other symmetry of the graph is a local Clifford, $U_G\subset \mathcal{C}_n$.}

\begin{proof}
We prove the statement by contradiction. Let $U\in U_G$ be such that $U_1\in\mathcal{C}_1^3$. Suppose there exists a $p\in\{2,\ldots,n\}$ such that $U_p\not\in \mathcal{C}_1^3$. We first show that this implies that there exists a symmetry $V\in U_G$ such that $V_p\propto\one$ and $V_1\in \mathcal{C}_1^3$ which, as we shown then, leads to a contradiction. For any $j\in\{2,\ldots, n\}$ for which $U_j\not\in \mathcal{C}_1^3$ we know due to Theorem \ref{generalform} that $U_j\propto \sigma_1^j\exp(i\alpha_j\sigma_2^j)$. Wlog let the qubits be ordered such that $U_j\in \mathcal{C}_1^3$ for $j\ge k+1$ for some $k\in\{1,\ldots,n-1\}$. Let us recursively define symmetries of $\ket{G}$ by $V(m+1)=(P(m)V(m))^2$, $V(1)=U$ and $m\in \N$, $m\le k$ where $P(m)\in\mathcal{S}_G$ is chosen such that $V(m+1)_m\propto\one$. Note that since the graph state is fully connected we can always find such a $P(m)$. By construction $V(k+1)_j\propto\one$ for all $j\le k$. Now observe that for any element $C\in\mathcal{C}_1^3$ also $\sigma C\in\mathcal{C}_1^3$ for all $\sigma\in \mathcal{P}_1$ and $C^2\in\mathcal{C}_1^3$. Thus, $V(k+1)_j\in\mathcal{C}_1^3$ for all $j>k$ and we choose $V=V(k+1)$. Hence, $V$ is a symmetry where each tensor factor is either in $\mathcal{C}_1^3$ or trivial.

Next let us show that $V$ has to act nontrivial on all qubits. Note that this already follows from Eq. (\ref{commute}) which implies that for a fully entangled graph state any LC symmetry of order $3$ is of the form $Q=((A,\one),(\one,A+1))$ in the binary representation and thus acts nontrivial on every qubit. Recall that $A$ is a diagonal matrix and the operator acting on qubit $j$ is given by $Q_j=((A_{jj},1),(1,A_{jj}+1)$.

In the following however, we provide a different proof which makes use of the fact that a conjugation of an element of $\mathcal{S}_G$ with $V$ has to give an element of the stabilizer again. The idea is to show that no qubit $j$ with $V_j\in\mathcal{C}_1^3$ can have a neighbour $l\in N_j$ with $V_l\propto \one$. Since the graph $G$ is fully connected this implies that $V_j\in C_1^3$ for all $j\in\{1,\ldots,n\}$. We show the claim again by contradiction. So suppose $V_j\in\mathcal{C}_1^3$ and there exists an $l\in N_j$ with $V_l\propto \one$. Wlog let us again order the qubits such that $j=1$ and $l=2$. Since $V$ is a local Clifford operator and $V\in U_G$ it has to map stabilizer operators to stabilizer operators under conjugation. Thus, we can decompose the image of a stabilizer operator under this conjugation into the canonical generators. By assumption $V_1\in\mathcal{C}_1^3$ and thus $V_1$ is a cyclic permutation of Pauli operators. The two options are $V_1ZV_1^\dagger\propto X$ and $V_1ZV_1^\dagger\propto Y$. In the first case the stabilizer $S_{(2)}=Z\otimes X\otimes \ldots$ is mapped to $S_{(2)}'=U_1S_{(2)}U_1^\dagger=X\otimes X\otimes\ldots$. Thus, the decomposition $S_{(2)}'=\prod_j S_{(j)}^{k_j}$, where $k_j\in\{0,1\}$ for all $j$, must contain the factor $S_{(1)}$. Note that $\text{supp}(S_{(2)}')=N_2\cup\{2\}$ and thus all neighbors of qubit $1$ have to be shared with qubit $2$ or be connected to neighbors of qubit $2$ (or both). Since qubit $1$ and $2$ are connected we need $S_{(2)}'$ to contain the canonical generators of an odd number of shared neighbors to get $X$ on qubit $1$ (and not $Y$). Since neighbors of $1$ not shared with $2$ are not in $\text{supp}(S_{(2)}')$ the operator $S_{(2)}'$ cannot contain contributions from their corresponding generators. Now consider $S_{(1)}=X\otimes Z\otimes\ldots$ which is mapped to $S_{(1)}'=U_1S_{(1)}U_1^\dagger=Y\otimes Z\otimes\ldots$. We have $\text{supp}(S_{(1)}')=N_1\cup\{1\}$ and thus the only neighbors of $2$ it contains are the ones which are shared with qubit $1$. Now we count the number of $Z$ operators contributed by the contained generators on qubit $2$. By the above considerations we get an odd number of $Z$ operators from the shared neighbors and one $Z$ since $S_{(1)}$ is contained in $S_{(1)}'$ and thus $(S_{(1)}')_2=Z^\text{2k}=\one$, $k\in \mathbb{N}$, which is a contradiction to the form of $S_{(1)}'$. This concludes first case. In the second case we have $U_1ZU_1^\dagger\propto Y$ and we obtain a contradiction by similar arguments.  

Finally, it remains to show that $U_G\subset \mathcal{C}_n$. This follows from Theorem \ref{generalform} and the fact that $V\in U_G$, as can be seen as follows. Suppose there was a $V'\in U_G$ and $V'\not\in \mathcal{C}_n$. Since $V'\not\in \mathcal{C}_n$ there exists a $p\in\{1,\ldots,n\}$ such that $V_p'\propto \sigma_1^p\exp(i\alpha_p\sigma_2^p)$ and $\alpha_p\neq m\pi/4$. Furthermore, also $V'V\in U_G$ but $(VV')_j$ is not of the form stated in Theorem \ref{generalform} which is a contradiction.
\end{proof}

\section{Proof to Observation \ref{char leaf}\label{appendixd}}
As stated in the main text we prove here Observation \ref{char leaf} which we restate here.\newline

\noindent {\bf Observation 2.}  \textit{Let $\ket{G}$ be a graph state on $n\ge 3$ qubits. Let qubit $1$ and $2$ be a leaf parent pair. Then
\begin{equation}
U=e^{i\alpha X}\otimes e^{-i\alpha Z}\otimes \one\ \ \ \alpha\in\mathbb{R}
\end{equation}
\indent is in $U_G$. Moreover, there exists no other unitary symmetry of the form $U_1\otimes U_2\otimes \one\in U_G$.}

\begin{proof}
Let $\ket{G}$ be a graph state and let qubit $1$ and $2$ be a leaf parent pair. Then we find that
\bea
e^{i\alpha X_1}\ket{G}&=&(\cos(\alpha)\one_1+i\sin(\alpha)X_1)\ket{G}\\
&=&(\cos(\alpha)\one+i\sin(\alpha)X_1(X_1\otimes Z_2))\ket{G}\\&=&e^{i\alpha Z_2}\ket{G}
\eea
holds for any $\alpha\in \mathbb{R}$ as $X_1\otimes Z_2\in \mathcal{S}_G$. Hence, any operator $\exp(i\alpha X_1)\otimes \exp(-i\alpha Z_1)$ is a local unitary symmetry of $\ket{G}$ which shows the first part of the claim. It remains to show that there exists no other local unitary symmetry that acts nontrivial only on qubit $1$ and $2$. Using Theorem \ref{generalform}, the fact that $U_G$ is a group and that $\exp(i\alpha X_1)\otimes \exp(-i\alpha Z_1)$ is a symmetry of $\ket{G}$ for any $\alpha\in\mathbb{R}$ we conclude that any symmetry of this type has to be of the form
\bea
\sigma_1^1e^{i\alpha_1 X_1}\otimes \sigma_1^2e^{i\alpha_2 Z_2},
\eea
where $\alpha_1,\alpha_2\in\mathbb{R}$. Let us multiply this symmetry from the right with the leaf symmetry $\exp(-i\alpha_1 X_1)\otimes \exp(i\alpha_1 Z_1)$. If $\sigma_1^2\not\propto\one$ we also multiply from the left with a $S\in\mathcal{S}_G$ such that $S_2=\sigma_2^2$. Squaring the result we obtain a new symmetry of $\ket{G}$ which is nontrivial only on qubit $2$. Observation \ref{2 factors} then implies that the symmetry has to be proportional to the identity and thus $\alpha_1+\alpha_2=k\pi/2$ for some $k\in\mathbb{Z}$. Hence, we conclude that any symmetry other than the leaf symmetry acting nontrivial only on qubit $1$ and $2$ has to be a 2-qubit Pauli operator and thus by Observation \ref{nootherpauli} an element of the stabilizer. Using that the stabilizer is abelian and $X_1\otimes Z_2\in \mathcal{S}_G$ we conclude that the only possibilities for additional symmetries are $Z_1\otimes X_2$ and $Y_1\otimes Y_2$ up to phases. Let us decompose both operators with respect to the canonical generators which gives $S_{(2)}$ and $S_{(1)}S_{(2)}$, respectively. However, as $n\ge 3$ and we consider connected graphs, qubit $2$ has to have a neighbour different from qubit $1$ and thus both operators act nontrivial on more qubits than just qubit $1$ and $2$. This completes the proof.
\end{proof}

\onecolumngrid
\section{Additional symmetries of the graph states presented in Figure \ref{graphs}\label{tablegraphssection}}
Table \ref{tablegraphs} contains a (not necessarily independent) set of generators for the additional symmetries of the graph states presented in Fig. \ref{graphs}.

\begin{table*}[h]
\centering
\begin{tabular}{||c||c|c|}\hline
state & \parbox[c][0.8 cm][c]{0.4\textwidth}{LC symmetries of order $3$}& \parbox[c][0.8 cm][c]{0.4\textwidth}{LC symmetries of order $4$ and $2^n$--th roots}
\\\hline\hline
\parbox[c][0.72 cm][c]{0.05\textwidth}{a)}
&
$(\exp(i\frac{\pi}{4}Z)\exp(i\frac{\pi}{4}Y))^{\otimes 4}\otimes(\exp(i\frac{\pi}{4}X)\exp(-i\frac{\pi}{4}Y))^{\otimes 4}$
& 
$\exp(i\frac{\pi}{4}Z)^{\otimes 4}\otimes\exp(i\frac{\pi}{4}X)^{\otimes 4}$, $\exp(i\frac{\pi}{4}Y)^{\otimes 4}\otimes\exp(-i\frac{\pi}{4}Y)^{\otimes 4}$
\\\hline
\parbox[c][0.72 cm][c]{0.05\textwidth}{b)}
&
$(\exp(-i\frac{\pi}{4}Z)\exp(-i\frac{\pi}{4}Y))^{\otimes 9}\otimes(\exp(i\frac{\pi}{4}X)\exp(-i\frac{\pi}{4}Y))^{\otimes 9}$
&
none
\\\hline
\parbox[c][0.72 cm][c]{0.05\textwidth}{c)}
&
none
&
$\exp(i\frac{\pi}{4}Z)^{\otimes 3}\otimes \exp(i\frac{\pi}{4}X)^{\otimes 4}$\\\hline
\parbox[c][0.72 cm][c]{0.05\textwidth}{d)}
&
none
&
$\exp(i\frac{\pi}{8} Z)^{\otimes 4}\otimes \exp(i\frac{\pi}{8} X)^{\otimes 11}$\\\hline
\parbox[c][0.72 cm][c]{0.05\textwidth}{e)}
&
none
&
$\exp(i\frac{\pi}{4} X)^{\otimes 4}\otimes \exp(-i\frac{\pi}{4} X)^{\otimes 4}$\\\hline
\parbox[c][0.72 cm][c]{0.05\textwidth}{f)}
&
none
&
$\one_{2,3}\otimes \bigotimes_{j\in\{1,4,5,7,9\}} \exp(i\frac{\pi}{4}X_j)\otimes \bigotimes_{j\in\{6,8\}} \exp(-i\frac{\pi}{4}X_j)$\\\hline
\end{tabular}
\caption{Additional generators needed to generate the full LU symmetry group of the graph states presented in Figure \ref{graphs}\label{tablegraphs}}
\end{table*}

\end{document}